\documentclass[acmsmall]{acmart}

\AtBeginDocument{%
  \providecommand\BibTeX{{%
    \normalfont B\kern-0.5em{\scshape i\kern-0.25em b}\kern-0.8em\TeX}}}

\setcopyright{rightsretained} 

\acmJournal{TQC}
\acmYear{2022} \acmVolume{3} \acmNumber{2} \acmArticle{5} \acmMonth{2} \acmPrice{}\acmDOI{10.1145/3498331}

\usepackage{graphicx}
\usepackage{dcolumn}
\usepackage{epstopdf}
\usepackage{bm}
\usepackage{subfig}
\usepackage{braket}
\usepackage{algorithm}
\usepackage{algorithmic}
\usepackage[T1]{fontenc}
\usepackage{color}
\usepackage{amsthm}
\usepackage{esint}
\usepackage{comment}
\usepackage{hyperref}
\usepackage{booktabs}

\newtheorem{thm}{\protect\theoremname}
  \theoremstyle{plain}
  \newtheorem{lem}[thm]{\protect\lemmaname}
  \theoremstyle{remark}
  \newtheorem{rem}[thm]{\protect\remarkname}
  \theoremstyle{plain}
  \newtheorem*{lem*}{\protect\lemmaname}
  \theoremstyle{plain}
  
  \theoremstyle{plain}
  \newtheorem{cor}[thm]{\protect\corollaryname}

  \usepackage[USenglish]{babel}
  \providecommand{\corollaryname}{Corollary}
  \providecommand{\lemmaname}{Lemma}
  \providecommand{\propositionname}{Proposition}
  \providecommand{\remarkname}{Remark}
\providecommand{\theoremname}{Theorem}

\begin{document}

\global\long\def\ve{\varepsilon}
\global\long\def\R{\mathbb{R}}
\global\long\def\Rn{{\mathbb{R}^{n}}}
\global\long\def\Rd{{\mathbb{R}^{d}}}
\global\long\def\E{\mathbb{E}}
\global\long\def\P{\mathbb{P}}
\global\long\def\bx{\mathbf{x}}
\global\long\def\vp{\varphi}
\global\long\def\ra{\rightarrow}
\global\long\def\smooth{C^{\infty}}
\global\long\def\symm{\mathcal{S}^n}
\global\long\def\psd{\mathcal{S}^n_{+}}
\global\long\def\pd{\mathcal{S}^n_{++}}
\global\long\def\dom{\mathrm{dom}\,}
\global\long\def\intdom{\mathrm{int}\,\mathrm{dom}\,}
\global\long\def\Tr{\mathrm{Tr}}

\newcommand{\LL}[1]{\textcolor{blue}{[LL:#1]}}
\newcommand{\DA}[1]{\textcolor{cyan}{[DA:#1]}}
\newcommand{\REV}[1]{\textcolor{red}{#1}}
\newcommand{\REVV}[1]{\textcolor{red}{#1}}

\newcommand{\bvec}[1]{\mathbf{#1}}
\renewcommand{\Re}{\mathrm{Re}}
\renewcommand{\Im}{\mathrm{Im}}
\newcommand{\textred}[1]{\textcolor{red}{#1}}

\newcommand{\mc}[1]{\mathcal{#1}}
\newcommand{\mf}[1]{\mathfrak{#1}}
\newcommand{\mcV}{\mathcal{V}}
\newcommand{\Vin}{V_{\mathrm{in}}}
\newcommand{\Vstar}{V^{\ast}}
\newcommand{\Jstar}{J_{\ast}}
\newcommand{\tJstar}{\wt{J}_{\ast}}
\newcommand{\Vout}{V_{\mathrm{out}}}
\newcommand{\RPA}{\mathrm{RPA}}
\newcommand{\xc}{\mathrm{xc}}
\newcommand{\vF}{\bvec{F}}
\newcommand{\vg}{\bvec{g}}
\newcommand{\vR}{\bvec{R}}
\newcommand{\vq}{\bvec{q}}
\newcommand{\vx}{\bvec{x}}
\newcommand{\ud}{\,\mathrm{d}}
\newcommand{\ext}{\mathrm{ext}}
\newcommand{\KS}{\mathrm{KS}}
\newcommand{\Exc}{E_{\mathrm{xc}}}
\newcommand{\Vxc}{\hat{V}_{\mathrm{xc}}}
\newcommand{\Vion}{\hat{V}_{\mathrm{ion}}}
\newcommand{\abs}[1]{\lvert#1\rvert}
\newcommand{\norm}[1]{\lVert#1\rVert}
\newcommand{\average}[1]{\left\langle#1\right\rangle}
\newcommand{\wt}[1]{\widetilde{#1}}
\newcommand{\hxc}{\mathrm{hxc}}

\newcommand{\etc}{\textit{etc.}~}
\newcommand{\etal}{\textit{et al}~}  
\newcommand{\ie}{\textit{i.e.}~}
\newcommand{\eg}{\textit{e.g.}~}
\newcommand{\Or}{\mathcal{O}}
\newcommand{\mcF}{\mathcal{F}}
\newcommand{\lmin}{\lambda_{\min}}
\newcommand{\lmax}{\lambda_{\max}}
\newcommand{\Ran}{\text{Ran}}
\newcommand{\I}{\mathrm{i}} 
\newcommand{\EE}{\mathbb{E}}
\newcommand{\NN}{\mathbb{N}}
\newcommand{\RR}{\mathbb{R}}
\newcommand{\CC}{\mathbb{C}}
\newcommand{\ZZ}{\mathbb{Z}}
\newcommand{\Hper}{H^1_\#(\Omega)}
\newcommand{\jmp}[1]{\jl#1\jr}
\newcommand{\al}{\{\hspace{-3.5pt}\{}
\newcommand{\ar}{\}\hspace{-3.5pt}\}}
\newcommand{\avg}[1]{\al#1\ar}
\newcommand{\jl}{[\![}
\newcommand{\jr}{]\!]}
\newcommand{\VN}{\mathbb V_N}
\newcommand{\angstrom}{\mbox{\normalfont\AA}~}

\title{Quantum linear system solver based on time-optimal adiabatic quantum computing and quantum approximate optimization algorithm}

\author{Dong An}
\email{dong\_an@berkeley.edu}
\affiliation{%
  \institution{Department of Mathematics, University of California, Berkeley}
  \city{Berkeley}
  \state{California}
  \country{USA}
  \postcode{94720}
}

\author{Lin Lin}
\email{linlin@math.berkeley.edu}
  \affiliation{%
  \institution{Department of Mathematics and Challenge Institute of Quantum Computation, University of California, Berkeley}
  \city{Berkeley}
  \state{California}
  \country{USA}
  \postcode{94720}}
  \affiliation{%
  \institution{Computational Research Division, Lawrence Berkeley National Laboratory}
  \city{Berkeley}
  \state{California}
  \country{USA}
  \postcode{94720}
}

\renewcommand{\shortauthors}{An and Lin}

\begin{abstract}
  We demonstrate that with an optimally tuned scheduling function,  adiabatic quantum computing (AQC) can  readily solve a quantum linear system problem (QLSP) with  $\Or(\kappa~\text{poly}(\log(\kappa/\epsilon)))$ runtime, where $\kappa$ is the condition number, and $\epsilon$ is the target accuracy.  This is near optimal with respect to both  $\kappa$ and $\epsilon$, and is achieved without relying on complicated amplitude amplification procedures that are difficult to implement. Our method is applicable to general non-Hermitian  matrices, and the cost as well as the number of qubits can be reduced when restricted to Hermitian matrices, and further to Hermitian positive definite matrices. The success of the time-optimal AQC implies that the quantum approximate optimization algorithm (QAOA) with an optimal control protocol can also achieve the same complexity in terms of the runtime. 
Numerical results indicate that QAOA can yield the lowest runtime compared to the time-optimal AQC, vanilla AQC, and the recently proposed randomization method.  
\end{abstract}



\maketitle

\section{Introduction} 
Linear system solvers are used ubiquitously in scientific computing. Quantum algorithms for solving large systems of linear equations, also called the quantum linear system problem (QLSP), have received much attention recently~\cite{HarrowHassidimLloyd2009,ChildsKothariSomma2017,ChakrabortyGilyenJeffery2018,GilyenSuLowEtAl2019,SubasiSommaOrsucci2019,WossnigZhaoPrakash2018,CaoPapageorgiouPetrasEtAl2013,XuSunEndoEtAl2019,Bravo-PrietoLaRoseCerezoEtAl2019}.  The goal of QLSP is to efficiently compute $\ket{x}=A^{-1}\ket{b}/\norm{A^{-1}\ket{b}}_2$ on a quantum computer, where $A\in\CC^{N\times N}$, and $\ket{b}\in\CC^N$ is a normalized vector  (for simplicity we assume $N=2^n$, and $\|A\|_2 = 1$). The ground-breaking Harrow, Hassidim, and Lloyd (HHL) algorithm obtains $\ket{x}$ with cost $\Or(\text{poly}(n) \kappa^2 /\epsilon)$, where $\kappa = \|A\|\|A^{-1}\|$ is the condition number of $A$, and $\epsilon$ is the target accuracy. On the other hand, the best classical iterative algorithm is achieved by the conjugate gradient method, where the cost is at least $\Or(N \sqrt{\kappa}\log(1/\epsilon))$, with the additional assumptions that $A$ should be Hermitian positive definite and a matrix-vector product can be done with $\Or(N)$ cost~\cite{Saad2003}. The complexity of direct methods based on the Gaussian elimination procedure removes the dependence on $\kappa$, but the dependence on $N$ is typically super-linear even for sparse matrices~\cite{Liu1992}. Therefore the HHL algorithm can potentially be exponentially faster than classical algorithms with respect to $N$. The undesirable dependence with respect to $\epsilon$ is due to the usage of the quantum phase estimation (QPE) algorithm. Recent progresses based on linear combination of unitaries (LCU)~\cite{ChildsKothariSomma2017} and quantum signal processing (QSP)~\cite{LowChuang2017,GilyenSuLowEtAl2019} have further improved the scaling to $\Or(\kappa^2 \text{poly}(\log(\kappa/\epsilon)))$ under different query models, without using QPE. However, the $\Or(\kappa^2)$ scaling can be rather intrinsic to the methods, at least before complex techniques such as variable time amplitude amplification (VTAA) algorithm~\cite{Ambainis2012} are applied. 

The VTAA algorithm is a generalization of the standard amplitude amplification algorithm, and allows to quadratically amplify the success probability of quantum algorithms in which different branches stop at different time. In~\cite{Ambainis2012}, VTAA is first used to successfully improve the complexity of HHL algorithm to $\widetilde{\Or}(\kappa/\epsilon^3)$. In~\cite{ChildsKothariSomma2017}, the authors further combine VTAA algorithm and a low-precision phase estimate to improve the complexity of LCU to $\widetilde{\Or}(\kappa~\text{poly}(\log(\kappa/\epsilon)))$, which is near-optimal with respect to both $\kappa$ and $\epsilon$. It is worth noting that the VTAA algorithm is a complicated procedure and can be difficult to implement. Thus, it remains of great interest to obtain alternative algorithms to solve QLSP with near-optimal complexity scaling without resorting to VTAA.

Some of the alternative routes for solving QLSP are provided by the adiabatic quantum computing (AQC)~\cite{JansenRuskaiSeiler2007,AlbashLidar2018} and a closely related method called the randomization method (RM)~\cite{BoixoKnillSomma2009,SubasiSommaOrsucci2019}. The key idea of both AQC and RM is to solve QLSP as an \textit{eigenvalue} problem with respect to a transformed matrix. Assume that a Hamiltonian simulation can be efficiently performed on a quantum computer, it is shown that the runtime of RM scales as $\Or(\kappa\log (\kappa)/\epsilon)$~\cite{SubasiSommaOrsucci2019}, which achieves near-optimal complexity with respect to $\kappa$ without using VTAA algorithm as a subroutine. The key idea of the RM is to approximately follow the adiabatic path based on the quantum Zeno effect (QZE) using a Monte Carlo method. Although RM is inspired by AQC, the runtime complexity of the (vanilla) AQC is at least $\Or(\kappa^2/\epsilon)$~\cite{SubasiSommaOrsucci2019,BoixoSomma2010,AlbashLidar2018}. Therefore the RM is found to be at least quadratically faster than AQC with respect to $\kappa$. 

In this paper, we find that with a simple modification of the scheduling function to traverse the adiabatic path, the gap between AQC and RM can be fully closed, along with the following two aspects. 1) We propose a family of rescheduled AQC algorithms called AQC(p). Assuming  $\kappa$ (or its upper bound) is known, we demonstrate that for any matrix $A$ (possibly non-Hermitian or dense), when $1<p<2$, the runtime complexity of AQC(p) can be only $\Or(\kappa/\epsilon)$. Thus AQC(p) removes a logarithmic factor with respect to $\kappa$ compared to RM. 2) We propose another rescheduled algorithm called AQC(exp), of which the runtime is  $\Or(\kappa~\text{poly}(\log(\kappa/\epsilon)))$. The main benefit of AQC(exp) is the improved dependence with respect to the accuracy $\epsilon$, and this is the near-optimal complexity (up to logarithmic factors)  with respect to both $\kappa$ and $\epsilon$.  The scheduling function of AQC(exp) is also universal because we do not even need the knowledge of an upper bound of $\kappa$. Existing works along this line \cite{Nenciu1993, GeMolnarCirac2016} only suggest that runtime complexity is $\Or(\kappa^3~\text{poly}(\log(\kappa/\epsilon)))$, which improves the dependence with respect to $\epsilon$ at the expense of a much weaker dependence on $\kappa$. Our main technical contribution is to again improve the dependence on $\kappa. $ 
Since the cost of any generic QLSP solver can not be less than $\Or(\kappa)$~\cite{HarrowHassidimLloyd2009}, our result achieves the near-optimal complexity up to logarithmic factors. 
We remark that in the AQC based algorithm, only the total runtime $T$ depends on $\kappa$. 

Beyond the runtime complexity, we also discuss gate-efficient approaches to implement our AQC(p) and AQC(exp) methods. 
In particular, assume that we are given access to the same query models as those in~\cite{ChildsKothariSomma2017}: the sparse input model of a $d$-sparse matrix $A$ and the prepare oracle of the state $\ket{b}$. 
We demonstrate that, when the adiabatic dynamics is simulated using the truncated Dyson series method~\cite{LowWiebe2019}, the query complexity of the AQC(p) method scales $\Or(d\kappa/\epsilon \log(d\kappa/\epsilon))$, and that of the AQC(exp) method scales $\Or(d\kappa~\text{poly}\log(d\kappa/\epsilon))$. 
Both algorithms scale almost linearly in terms of $\kappa$, and the AQC(exp) method can achieve near-optimal scaling in both $\kappa$ and $\epsilon$. 
Furthermore, the asymptotic scaling of the AQC(exp) method is the same as that of LCU with VTAA method~\cite[Theorem 5]{ChildsKothariSomma2017}. However, the AQC(exp) method avoids the usage of complex VTAA routine, which significantly simplifies its practical implementation. 

The quantum approximate optimization algorithm (QAOA)~\cite{FarhiGoldstoneGutmann2014}, as a quantum variational algorithm, has received much attention recently thanks to the feasibility of being implemented on near-term quantum devices. Due to the natural connection between AQC and QAOA, our result immediately suggests that the time-complexity for solving QLSP with QAOA is also at most $\Or(\kappa~\text{poly}(\log(\kappa/\epsilon)))$, which is also confirmed by numerical results. We also remark that both QAOA and AQC schemes prepare an approximate solution to the QLSP in a pure state, while RM prepares a mixed state. Moreover, all methods above can be efficiently implemented on gate-based computers and are much simpler than those using the VTAA algorithm as a subroutine. 

\section{Quantum Linear System Problem and Vanilla AQC}
Assume $A\in \CC^{N\times N}$ is an invertible matrix with condition number $\kappa$ and $\|A\|_2 = 1$. 
Let $\ket{b}\in\CC^{N}$ be a normalized vector. 
Given a target error $\epsilon$, the goal of QLSP is to prepare a normalized state 
$\ket{x_{\text{a}}}$, which is an 
$\epsilon$-approximation of the normalized solution of the linear system $\ket{x}=A^{-1}\ket{b}/\norm{A^{-1}\ket{b}}_2$, 
in the sense that $\|\ket{x_{\text{a}}}\bra{x_{\text{a}}}-\ket{x}\bra{x}\|_2 \le \epsilon$. 

For simplicity, we  first assume $A$ is Hermitian and positive definite and will discuss the generalization to non-Hermitian case later. 

The first step to design an AQC-based algorithm for solving QLSP is to transform the QLSP to an equivalent eigenvalue problem. Here we follow the procedure introduced in~\cite{SubasiSommaOrsucci2019}. 
Let $Q_{b}=I_N-\ket{b}\bra{b}$. 
We introduce
$$
H_0=\sigma_x \otimes Q_b=\begin{pmatrix}
0 & Q_b\\
Q_b & 0
\end{pmatrix},
$$
then $H_0$ is a Hermitian matrix and the null space of $H_0$ is $\text{Null}(H_0)=\text{span}\{\ket{\wt{b}},\ket{\bar{b}}\}$. Here $\ket{\wt{b}}=\ket{0,b}:=(b,0)^{\top},\ket{\bar{b}}=\ket{1,b}:=(0,b)^{\top}$. The dimension of $H_0$ is $2N$ and one ancilla qubit is needed to enlarge the matrix block. We also define
$$
H_1=\sigma_{+}\otimes (AQ_b)+\sigma_{-}\otimes (Q_bA)=\begin{pmatrix}
0 & AQ_b\\
Q_bA & 0
\end{pmatrix}.
$$
Here $\sigma_{\pm}=\frac12(\sigma_x\pm \I\sigma_y)$. Note that if $\ket{x}$ satisfies $A\ket{x}\propto \ket{b}$, we have $Q_bA\ket{x}=Q_b\ket{b}=0$. Then  $\text{Null}(H_1)=\text{span}\{\ket{\wt{x}},\ket{\bar{b}}\}$ with $\ket{\wt{x}}=\ket{0,x}$. 
Since $Q_b$ is a projection operator, the gap between $0$ and the rest of the eigenvalues of $H_0$ is $1$. The gap between $0$ and the rest of the eigenvalues of $H_1$ is bounded from below  
by $1/\kappa$ (see Appendix~\ref{app:gap}).

QLSP can be solved if we can prepare the zero-energy state $\ket{\wt{x}}$ of $H_1$, which can be achieved by the AQC approach. 
Let $H(f(s)) = (1-f(s))H_0 + f(s)H_1, 0\le s\le 1$. 
The function $f:[0,1]\rightarrow [0,1]$ is called a scheduling function, and is a strictly increasing mapping with $f(0) = 0, f(1) = 1$. 
The simplest choice is $f(s)=s$, which gives the ``vanilla AQC''.  
We sometimes omit the $s$-dependence as $H(f)$ to emphasize the dependence 
on $f$. 
Note that for any $s$, $\ket{\bar{b}}$ is always in $\text{Null}(H(f(s)))$,  and there exists a state $\ket{\wt{x}(s)}=\ket{0,x(s)}$, such that $\text{Null}(H(f(s)))=\{\ket{\wt{x}(s)},\ket{\bar{b}}\}$. In particular, $\ket{\wt{x}(0)}=\ket{\wt{b}}$ and $\ket{\wt{x}(1)}=\ket{\wt{x}}$, and therefore $\ket{\wt{x}(s)}$ is the desired adiabatic path. Let $P_0(s)$ be the projection to the subspace $\text{Null}(H(f(s)))$, which is a rank-2 projection operator $P_0(s)=\ket{\wt{x}(s)}\bra{\wt{x}(s)}+\ket{\bar{b}}\bra{\bar{b}}$.
Furthermore, the eigenvalue $0$ is separated from the rest of the eigenvalues of $H(f(s))$ by a gap 
\begin{equation}\label{eqn:gap_pd}
    \Delta(f(s))\ge \Delta_*(f(s)) := 1-f(s)+f(s)/\kappa. 
\end{equation}
We refer to Appendix~\ref{app:gap} for the derivation. 

Consider the adiabatic evolution
\begin{equation}    
  \frac{1}{T}\I \partial_s \left|\psi_T(s)\right> = H(f(s))\left|\psi_T(s)\right>, \quad     \ket{\psi_T(0)}=\ket{\wt{b}}, 
\label{eqn:adiabatic}
\end{equation}
where $0 \leq s \leq 1$, and the parameter $T$ is called the runtime of AQC. 
The quantum adiabatic theorem~\cite[Theorem 3]{JansenRuskaiSeiler2007} states that for any $0\le s\le 1$,
\begin{equation}
  |1-\braket{\psi_T(s)|P_0(s)|\psi_T(s)}|\le \eta^2(s),
  \label{eqn:adiabaticEstimate}
\end{equation}
where 
\begin{equation}\label{eqn:adiabaticEstimate_eta}
  \eta(s)=C\Big\{\frac{\|H^{(1)}(0)\|_2}{T \Delta^2(0)} + \frac{\|H^{(1)}(s)\|_2}{T \Delta^2(f(s))} + \frac{1}{T}\int_0^s \left(\frac{\|H^{(2)}(s')\|_2}{\Delta^2(f(s'))} + \frac{\|H^{(1)}(s')\|^2_2}{\Delta^3(f(s'))}\right)ds'\Big\}.
\end{equation}
The derivatives of $H$ are taken with respect to $s$, \ie 
$H^{(k)}(s) := \frac{d^k}{ds^k} H(f(s)), k = 1,2$. 
Throughout the paper, we shall use a generic symbol $C$ to denote constants independent of $s,\Delta,T$.

Intuitively, the quantum adiabatic theorem in Eq.~\eqref{eqn:adiabaticEstimate} says that, if the initial state is an eigenstate corresponding to the eigenvalue 0, then for large enough $T$ the state $\ket{\psi_T(s)}$ will almost stay in the eigenspace of $H(s)$ corresponding to the eigenvalue 0, where there is a double degeneracy and only one of the eigenstate $\ket{\wt{x}(s)}$ is on the desired adiabatic path. 
However, such degeneracy will not break the effectiveness of AQC for the following reason. 
Note that $\braket{\bar{b}|\psi_T(0)}=0$, and $H(f(s))\ket{\bar{b}}=0$ for all $0\le s\le 1$, so the Schr\"odinger dynamics~\eqref{eqn:adiabatic} implies $\braket{\bar{b}|\psi_T(s)}=0$, which prevents any transition of $\ket{\psi_T(s)}$ to $\ket{\bar{b}}$. 
Therefore the adiabatic path will stay along $\ket{\wt{x}(s)}$. Using $\braket{\bar{b}|\psi_T(s)}=0$, we have $P_0(s)\ket{\psi_T(s)}=\ket{\wt{x}(s)}\braket{\wt{x}(s)|\psi_T(s)}$.
Therefore the estimate in Equation~\eqref{eqn:adiabaticEstimate} becomes
$$
1-|\braket{\psi_T(s)|\wt{x}(s)}|^2\le \eta^2(s).
$$
This also implies that (see Appendix~\ref{app:measurements}) 
$$
\norm{\ket{\psi_T(s)}\bra{\psi_T(s)}-\ket{\wt{x}(s)}\bra{\wt{x}(s)}}_2\le \eta(s).
$$
Therefore $\eta(1)$ can be an upper bound of the distance of the density matrix. 

If we simply assume $\norm{H^{(1)}}_2,\norm{H^{(2)}}_2$ are bounded by constants, and use the worst case bound that $\Delta\ge \kappa^{-1}$, we arrive at the conclusion that in order to have $\eta(1)\le \epsilon$, the runtime of vanilla AQC is $T\gtrsim \kappa^3/\epsilon$.  

\section{AQC(p) method}
Our goal is to reduce the runtime by choosing a proper scheduling function. 
The key observation is that the accuracy of AQC depends not only on 
the gap $\Delta(f(s))$ but also on the derivatives of $H(f(s))$, as revealed in the estimate in Equation~\eqref{eqn:adiabaticEstimate_eta}. 
Therefore it is possible to improve the accuracy if a proper time schedule allows 
the Hamiltonian $H(f(s))$ to slow down when the gap is close to $0$. 
We consider the following 
schedule~\cite{JansenRuskaiSeiler2007,AlbashLidar2018} 
\begin{equation}
    \dot{f}(s) = c_p \Delta_*^p(f(s)), \quad f(0) = 0, \quad p > 0{.}
    \label{eqn:AQCSchedule}
\end{equation}
Here $\Delta_*$ is defined in Eq.~\eqref{eqn:gap_pd} and $c_p = \int_0^1 \Delta_*^{-p}(u) du$ is a normalization constant chosen so that $f(1) = 1$.  When $1 < p \leq 2$, Eq.~\eqref{eqn:AQCSchedule} can be explicitly solved as 
\begin{equation}
    f(s) = \frac{\kappa}{\kappa - 1}\left[1-\left(1+s(\kappa^{p-1}-1)\right)^{\frac{1}{1-p}}\right]{.}
    \label{eqn:AQCSchedule_explicit}
\end{equation}
Note that as $s\to 1$, $\Delta_{*}(f(s))\to \kappa^{-1}$, and therefore the dynamics of $f(s)$ slows down as $f\to 1$ and the gap decreases towards $\kappa^{-1}$. We refer to the adiabatic dynamics (Equation~\eqref{eqn:adiabatic}) with the schedule in Equation~\eqref{eqn:AQCSchedule} as the AQC(p) scheme. Our main result is given in Theorem~\ref{thm:main} (See Appendix~\ref{app:proof_linear} for the proof). 

\begin{thm}\label{thm:main}
Let $A\in\CC^{N\times N}$ be a Hermitian positive definite matrix with condition number $\kappa$. For any choice of $1 < p < 2$, the error of the AQC(p) scheme satisfies 
\begin{equation}
    \|\ket{\psi_T(1)}\bra{\psi_T(1)}-\ket{\wt{x}}\bra{\wt{x}}\|_2 \leq C \kappa/T.
\end{equation}
Therefore in order to prepare an $\epsilon-$approximation of the solution of QLSP it suffices to choose the runtime  $T = \Or(\kappa/\epsilon)$. Furthermore, when $p=1,2$, the bound for the runtime becomes $T=\Or(\kappa\log(\kappa)/\epsilon)$.
\end{thm}

The runtime complexity of the AQC(p) method with respect to $\kappa$ is only $\Or(\kappa)$. Compared to Ref.~\cite{SubasiSommaOrsucci2019}, AQC(p) further removes the $\log(\kappa)$ dependence when $1<p<2$, and hence reaches the optimal complexity with respect to $\kappa$. Interestingly, though not explicitly mentioned in \cite{SubasiSommaOrsucci2019}, the success of RM for solving QLSP relies on a proper choice of the scheduling function, which approximately corresponds to AQC(p=1). It is this scheduling function, rather than the QZE or its Monte Carlo approximation \textit{per se} that achieves the desired $\Or(\kappa\log\kappa)$ scaling with respect to $\kappa$. Furthermore, the scheduling function in RM is similar to the choice of the schedule in the AQC(p=1) scheme. The speedup of AQC(p) versus the vanilla AQC is closely related to the quadratic speedup of the optimal time complexity of AQC for Grover's search~\cite{RolandCerf2002,JansenRuskaiSeiler2007,RezakhaniKuoHammaEtAl2009,AlbashLidar2018}, in which the optimal time scheduling reduces the runtime from $T\sim \Or(N)$ (i.e. no speedup compared to classical algorithms) to $T\sim \Or(\sqrt{N})$ (i.e. Grover speedup). In fact, the choice of the scheduling function in Ref.~\cite{RolandCerf2002} corresponds to AQC(p=2) and that in Ref.~\cite{JansenRuskaiSeiler2007} corresponds to AQC(1<p<2).

\section{AQC(exp) method}\label{sec:aqcexp}
Although AQC(p) achieves the optimal runtime complexity with respect to $\kappa$, the dependence on $\epsilon$ is still $\Or(\epsilon^{-1})$, which limits the method from achieving high accuracy. It turns out that when $T$ is sufficiently large, the dependence on $\epsilon$ could be improved to 
$\Or(\text{poly} \log (1/\epsilon))$, by choosing an alternative scheduling function.

The basic observation is as follows. In the AQC(p) method, the adiabatic error bound we consider, \ie Eq.~\eqref{eqn:adiabaticEstimate_eta}, is the so-called instantaneous adiabatic error bound, which holds true for all $s\in[0,1]$. However, when using AQC for solving QLSP, it suffices only to focus on the  error bound at the final time $s=1$. It turns out that this allows us to obtain a tighter error bound. In fact, such an error bound can be exponentially small with respect to the runtime~\cite{Nenciu1993,WiebeBabcock2012,GeMolnarCirac2016,AlbashLidar2018}. Roughly speaking, with an additional assumption for the Hamiltonian $H(f(s))$ that the derivatives of any order vanish at $s=0,1$, the adiabatic error can be bounded by $c_1\exp(-c_2T^{\alpha})$ for some positive constants $c_1, c_2, \alpha$. 
Furthermore, it is proved in~\cite{GeMolnarCirac2016} that if the target eigenvalue is simple, then $c_1 = \Or(\Delta_*^{-1})$ and $c_2 = \Or(\Delta_*^3)$. Note that $\Delta_* \geq \kappa^{-1}$ for QLSP, 
thus, according to this bound, to obtain an $\epsilon$-approximation, it suffices to choose $T = \Or(\kappa^3~\text{poly}(\log(\kappa/\epsilon)))$. This is an exponential speedup with respect to $\epsilon$, but the dependence on the condition number becomes cubic again.

However, it is possible to reduce the runtime if the change of the Hamiltonian is slow when the gap is small, as we have already seen in the AQC(p) method.  For QLSP, the gap monotonically decreases, and the smallest gap occurs uniquely at the final time, where the Hamiltonian $H(s)$ can be set to vary slowly by requiring its derivatives to vanish at the boundary. 

We consider the following schedule
\begin{equation}\label{eqn:AQCexpSchedule}
    f(s) = c_e^{-1} \int_0^s \exp\left(-\frac{1}{s'(1-s')}\right)\ud s'
\end{equation}
where $c_e = \int_0^1\exp\left(-1/(s'(1-s'))\right)\ud s'$ is a normalization constant such that $f(1) = 1$. 
This schedule can assure that $H^{(k)}(0) = H^{(k)}(1) = 0$ for all $k\geq 1$. 
We refer to the adiabatic dynamics (Equation~\eqref{eqn:adiabatic}) with the schedule in Equation~\eqref{eqn:AQCexpSchedule} as the AQC(exp) scheme.  
Our main result is given in Theorem~\ref{thm:main_exp} (see Appendix~\ref{app:exp} for the proof). 
\begin{thm}\label{thm:main_exp}
    Let $A\in\CC^{N\times N}$ be a Hermitian positive definite matrix with condition number $\kappa$. Then for large enough $T>0$, the final time error 
    $\|\ket{\psi_T(1)}\bra{\psi_T(1)} - \ket{\wt{x}}\bra{\wt{x}}\|_2$ 
    of the AQC(exp) scheme is bounded by 
\begin{equation}
    C\log(\kappa)\exp\left(-C\left(\frac{\kappa\log^2\kappa}{T}\right)^{-\frac{1}{4}}\right).
\end{equation}
Therefore for any $\kappa>e$, $0<\epsilon<1$, in order to prepare an $\epsilon-$approximation of the solution of QLSP, it suffices to choose the runtime $T = \Or\left(\kappa\log^2(\kappa)\log^4\left(\frac{\log\kappa}{\epsilon}\right)\right)$. 
\end{thm}

Compared with RM and AQC(p), although the $\log(\kappa)$ dependence reoccurs, AQC(exp) achieves an exponential speedup over RM and AQC(p) with respect to $\epsilon$ (and hence giving its name), and thus is more suitable for preparing the solution of QLSP with high fidelity. Furthermore, the time scheduling of AQC(exp) is 
universal and AQC(exp) does not require knowledge on the bound of $\kappa$. 

We remark that the performance of the AQC(exp) method is sensitive to the perturbations in the scheduling function, which can affect the final error in the AQC(exp) method. 
This is similar to the finite precision effect in the adiabatic Grover search reported in~\cite{Hen2019}. 
Therefore the scheduling function should be computed to sufficient accuracy on classical computers using numerical quadrature, and implemented accurately as a control protocol on quantum computers.

\section{Gate-based implementation of AQC}\label{sec:aqc_implementation}

We briefly discuss how to implement AQC(p) and AQC(exp) on a gate-based quantum computer. Since $\ket{\psi_T(s)} = \mathcal{T}\exp(-\I T\int_0^sH(f(s'))ds')\ket{\psi_T(0)}$, where $\mathcal{T}$ is the time-ordering operator, it is sufficient to implement an efficient time-dependent Hamiltonian simulation 
of $H(f(s))$. 

One straightforward approach is to use the Trotter splitting method. The lowest order approximation takes the form
\begin{equation}\label{eqn:splitting}
\begin{split}
    \mathcal{T}\exp\left(-\I T\int_0^sH(f(s'))\ud s'\right) \approx \prod_{m=1}^M \exp\left(-\I T h H(f(s_m))\right)\\
    \approx \prod_{m=1}^M \exp\left(-\I T h (1-f(s_m))H_0\right)\exp\left(-\I T h f(s_m)H_1\right)
\end{split}    
\end{equation}
where $h = s/M, s_m = mh$. It is proved in~\cite{vanDamMoscaVazirani2001} that the error of such an approximation is \\
$\Or(\text{poly}(\log (N))T^2/M)$, which indicates that to achieve an  $\epsilon$-approximation, it suffices to choose $M = \Or(\text{poly}(\log (N))T^2/\epsilon)$. 
On a quantum computer, the operations $e^{-\I\tau H_0},e^{-\I\tau H_1}$ require a time-independent Hamiltonian simulation process, which can be implemented via techniques such as LCU and QSP~\cite{BerryChildsCleveETC2015,LowChuang2017}. 
For a $d$-sparse matrix $A$, according to~\cite{BerryChildsKothari2015}, the query complexity is $\widetilde{\Or}(d\tau \log(d\tau/\epsilon))$ for a single step. Here $f=\widetilde{\mathcal{O}}(g)$ if $f=\mathcal{O}(g~\text{poly}\log (g))$. Note that the total sum of the simulation time of single steps is exactly $T$ regardless of the choice of $M$, and the total query complexity is $\widetilde{\Or}(dT\log(dT/\epsilon))$. Using Theorem~\ref{thm:main} and~\ref{thm:main_exp}, the query complexity of AQC(p) and AQC(exp) is $\widetilde{\Or}(d\kappa/\epsilon\log(d\kappa/\epsilon))$ and $\widetilde{\Or}(d\kappa~\text{poly}(\log(d\kappa/\epsilon)))$, respectively. 
Nevertheless, $M$ scales as $\Or(T^2)$ with respect to the runtime $T$, which implies that the number of time slices should be at least 
$\Or(\kappa^2)$. Therefore the gate complexity scales superlinearly with respect to $\kappa$. The scaling of the Trotter expansion can be improved using high order Trotter-Suzuki formula as well as the recently developed commutator-based error analysis \cite{ChildsSuTranEtAl2019}, but we will not pursue this direction here.

There is an efficient way to directly perform the time evolution of $H(f(s))$ without using the splitting strategy, following the algorithm proposed by Low and Wiebe in~\cite{LowWiebe2019}, where the time-dependent Hamiltonian simulation is performed based on a truncated Dyson expansion.
A detailed discussion on how to implement this algorithm in a gate-efficient way is presented in~\cite[Appendix C]{LinTong2019}, and here we summarize the basic idea as follows. Assume that we are given two input query models: $\mathcal{P}_A$ that gives the locations and values of the nonzero entries of the matrix $A$, and $\mathcal{P}_B$ that produces the quantum state $\ket{b}$. Here the input query models are the same as those in~\cite{ChildsKothariSomma2017}. Then one can construct the block-encoding representations of the matrix $A$~\cite{GilyenSuLowEtAl2019} and the matrix $Q_b$ with $\Or(1)$ additional primitive gates. Next, the block-encodings of $A$ and $Q_b$ can be applied to build the block-encodings of $H_0$ and $H_1$, and then the HAM-T model, which is a block-encoding of the select oracle of the time-dependent Hamiltonian $H(s)$ evaluated at different time steps and serves as the input model in the truncated Dyson series  method~\cite{LowWiebe2019}. Finally, after the construction of HAM-T, the adiabatic dynamics can be simulated following the procedure for solving time-dependent Schr\"odinger equations discussed in~\cite{LowWiebe2019}. 

The costs of AQC(p) and AQC(exp) are summarized in Table~\ref{tab:implementation_scaling}, where for both AQC(p) and AQC(exp), almost linear dependence with respect to $\kappa$ is achieved. 
The almost linear dependence on $\kappa$ cannot be expected to be improved to $\Or(\kappa^{1-\delta})$ with any $\delta>0$~\cite{HarrowHassidimLloyd2009}. Thus both AQC(p) and AQC(exp) are almost optimal with respect to $\kappa$, and AQC(exp) further achieves an exponential speedup with respect to $\epsilon$.

\begin{table}[]
    \centering
    \begin{tabular}{c|c|c}
         & AQC(p) & AQC(exp) \\\hline
        Queries & $\Or(d\kappa/\epsilon\log(d\kappa/\epsilon))$ & $\Or(d\kappa\text{ poly}(\log(d\kappa/\epsilon)))$ \\
        Qubits & $\Or(n+\log(d\kappa/\epsilon))$ & $\widetilde{\Or}(n+\log(d\kappa/\epsilon))$ \\
        Primitive gates & $\Or(nd\kappa/\epsilon~\text{poly}(\log(d\kappa/\epsilon)))$ & $\Or(nd\kappa~ \text{poly}(\log(d\kappa/\epsilon)))$ \\
    \end{tabular}
    \caption{Computational costs of AQC(p) and AQC(exp) via a time-dependent Hamiltonian simulation using the truncated Dyson expansion \cite{LowWiebe2019}. }
    \label{tab:implementation_scaling}
\end{table}

\section{QAOA for solving QLSP}
The quantum approximate optimization algorithm (QAOA)~\cite{FarhiGoldstoneGutmann2014} considers the following parameterized wavefunction
\begin{equation}
\ket{\psi_{\theta}}:=e^{-\I\gamma_P H_1} e^{-\I\beta_P H_0}\cdots e^{-\I\gamma_1 H_1} e^{-\I\beta_1 H_0}\ket{\psi_i}.
\label{eqn:QAOA}
\end{equation}
Here $\theta$ denotes the set of $2P$ adjustable real parameters $\{\beta_i,\gamma_i\}_{i=1}^{P}$, and $\ket{\psi_i}$ is an initial wavefunction. The goal of QAOA is to choose $\ket{\psi_i}$ and to tune  $\theta$, so that $\ket{\psi_\theta}$ approximates a target state. 
In the context of QLSP, we may choose $\ket{\psi_i}=\ket{\wt{b}}$, 
and each step of the QAOA ansatz in Eq.~\eqref{eqn:QAOA} can be efficiently implemented using the quantum singular value transformation~\cite{GilyenSuLowEtAl2019}. 
More specifically, as discussed in Section~\ref{sec:aqc_implementation} and in~\cite{LinTong2019}, the block-encodings of $H_0$ and $H_1$ can be efficiently constructed via the input models for the matrix $A$ and the vector $\ket{b}$. Then the quantum singular value transformation can be directly applied to simulate $e^{-\I \beta H_0}$ and $e^{-\I \gamma H_1}$. 
According to~\cite[Corollary 62]{GilyenSuLowEtAl2019}, the cost of each single simulation scales linearly in time and logarithmically in precision, and hence the total complexity of implementing a QAOA ansatz scales linearly in total runtime of QAOA, defined to be $T:=\sum_{i=1}^P(|\beta_i|+|\gamma_i|)$, and logarithmically in precision.
Notice that with a sufficiently large $P$, the optimal Trotter splitting method becomes a special form of Eq. \eqref{eqn:QAOA}. 
Hence Theorem~\ref{thm:main_exp} implies that with an optimal choice of $\{\beta_i,\gamma_i\}_{i=1}^P$, the QAOA runtime $T$ is at most $\Or(\kappa~\text{poly}(\log(\kappa/\epsilon)))$.
We remark that the validity of such an upper bound requires a sufficiently large $P$ and an optimal choice of $\theta$. On the other hand, our numerical results suggest that the same scaling can be achieved with a much smaller  $P$.

For a given $P$, the optimal $\theta$ maximizes the fidelity as
$$
\max_{\theta} F_{\theta}:=|\braket{\psi_{\theta}|\wt{x}}|^2.
$$
However, the maximization of the fidelity requires the knowledge of the exact solution $\ket{\wt{x}}$ which is not practical. We may instead solve the following minimization problem
\begin{equation}
\min_\theta     \braket{\psi_\theta|H_1^2|\psi_\theta}.
\label{eqn:minQAOA}
\end{equation}
Since the null space of $H_1$ is of dimension 2, the unconstrained minimizer $\ket{\psi_{\theta}}$ seems possible to only have a small overlap with $\ket{\wt{x}}$. However, 
this is not a problem due to the choice of the initial state $\ket{\psi_i}=\ket{\wt{b}}$. Notice that by the variational principle the minimizer $\ket{\psi_\theta}$ maximizes $\braket{\psi_\theta|P_0(1)|\psi_\theta}$. 
Using the fact that $e^{-\I \beta H_0}\ket{\bar{b}}=e^{-\I \gamma H_1}\ket{\bar{b}}=\ket{\bar{b}}$ for any $\beta,\gamma$, we obtain 
$\braket{\bar{b}|\psi_\theta}=\braket{\bar{b}|\wt{b}}=0,$
which means the QAOA ansatz prevents the transition to $\ket{\Bar{b}}$, similar to AQC. 
Then $\braket{\psi_\theta|P_0(1)|\psi_\theta}=\braket{\psi_\theta|\wt{x}}\braket{\wt{x}|\psi_\theta}=F_\theta$, so the minimizer of Eq. $\eqref{eqn:minQAOA}$ indeed maximizes the fidelity.  

For every choice of $\theta$, we evaluate the expectation value $\braket{\psi_\theta|H_1^2|\psi_\theta}$. Then the next $\theta$ is adjusted on a classical computer towards minimizing the objective function. The process is repeated till convergence. Efficient classical algorithms for the optimization of parameters in QAOA are currently an active topic of research, including methods using gradient optimization~\cite{ZhuRabitz1998,MadayTurinici2003},  Pontryagin's maximum principle (PMP)~\cite{YangRahmaniShabaniEtAl2017,BaoKleerWangEtAl2018}, reinforcement learning~\cite{BukovDaySelsEtAl2018,NiuBoixoSmelyanskiyEtAl2019}, to name a few. Algorithm~\ref{alg:qaoa} describes the procedure using QAOA to solve QLSP.

\begin{algorithm}[H]
\caption{QAOA for solving QLSP}
\label{alg:qaoa}
\begin{algorithmic}[1]
\STATE Initial parameters $\theta^{(0)}=\{\beta_i,\gamma_i\}_{i=1}^{2P}$.
\FOR{$k=0,1,\ldots$}
\STATE Perform Hamiltonian simulation to obtain $\psi_\theta^{(k)}$.
\STATE Measure 
$O(\theta^{(k)})=\braket{\psi^{(k)}_\theta|H_1^2|\psi^{(k)}_\theta}$.
\STATE If $O(\theta^{(k)})<\epsilon^2/\kappa^2$, exit the loop.
\STATE Choose $\theta^{(k+1)}$ using a classical optimization method.
\ENDFOR
\end{algorithmic}
\end{algorithm}

Compared to AQC(p) and AQC(exp), QAOA has the following two potential advantages. 
The first advantage is that QAOA provides the possibility of going beyond  AQC-based algorithms. 
Notice that the Trotter splitting method is a special form of the QAOA ansatz in Eq.~\eqref{eqn:QAOA}. 
If the angles $\{\beta_i,\gamma_i\}_{i=1}^P$ have been properly optimized (which is a very strong assumption and will be further discussed later), the total QAOA runtime $T$ will be by definition comparable to or even shorter than the runtime of AQC with the best scheduling function (after discretization).
Second, one way of implementing AQC(p) and AQC(exp) using an operator splitting method requires the time interval to be explicitly split into a large number of intervals, while numerical results indicate that the number of intervals $P$ in QAOA can be much smaller. This could reduce the depth of the quantum circuit. 
Compared to AQC, QAOA has the additional advantage that it only consists of $2P$ \textit{time-independent} Hamiltonian simulation problem, once $\theta$ is known.

Despite the potential advantages, several severe caveats of using QAOA for QLSP arise when we consider beyond the time complexity. 
The first is that classical optimization of the angles $\{\beta_i,\gamma_i\}_{i=1}^{P}$ can be difficult. 
Commonly used classical optimization algorithms, such as the gradient descent method, are likely to be stuck at local optimizers and thus result in sub-optimal performance. 
The cost for the classical optimization is also hard to known \textit{a priori}.
The optimization may require many iterations, which can diminish the gain of the runtime reduction. 
The second is related to the accurate computation of the objective function $O(\theta^{(k)})$. 
Note that the minimal spectrum gap of $H_1$ is $\Or(\kappa^{-1})$. In order to obtain an $\epsilon$-approximation, the precision of measuring $O(\theta)=\braket{\psi_\theta|H_1^2|\psi_\theta}$ should be at least $\Or(\epsilon^2/\kappa^2)$. 
Hence $\Or(\kappa^4/\epsilon^4)$ repeated measurements can be needed to achieve the desired accuracy.

\section{Generalization to non-Hermitian matrices}\label{sec:nonhermitian}

Now we discuss the case when $A$ is not Hermitian positive definite. 
First, we still assume that $A$ is Hermitian (but not necessarily positive 
definite). In this case, we adopt the family of Hamiltonians introduced 
in~\cite{SubasiSommaOrsucci2019}, which overcomes the difficulty 
brought by the indefiniteness of $A$ at the expense of enlarging 
the Hilbert space to dimension $4N$ (so two ancilla qubits are needed to enlarge the matrix block).  Here we define
$$
H_0=\sigma_+ \otimes \left[(\sigma_z \otimes I_N)Q_{+,b}\right] + \sigma_- \otimes \left[Q_{+,b}(\sigma_z \otimes I_N)\right]
$$
where $Q_{+,b} = I_{2N}-\ket{+,b}\bra{+,b}$, and $\ket{\pm}=\frac{1}{\sqrt{2}}(\ket{0}\pm\ket{1})$.  The null space of $H_0$ is $\text{Null}(H_0)=\text{span}\{\ket{0,-,b}, \ket{1,+,b}\}$. 
We also define
$$
H_1=\sigma_+ \otimes \left[(\sigma_x \otimes A)Q_{+,b}\right] + \sigma_- \otimes \left[Q_{+,b}(\sigma_x \otimes A)\right]
$$
Note that $\text{Null}(H_1)=\text{span}\{\ket{0,+,x},\ket{1,+,b}\}$. Therefore the solution of the 
QLSP can be obtained if we can prepare the zero-energy state 
$\ket{0,+,x}$ of $H_1$. 

The family of Hamiltonians for AQC(p) is still given by
$H(f(s)) = (1-f(s))H_0 + f(s)H_1, 0\le s\le 1$. 
Similar to the case of Hermitian positive definite matrices, 
there is a double degeneracy of the eigenvalue $0$, and we aim at preparing one of 
the eigenstate via time-optimal adiabatic evolution. 
More precisely, for any $s$, $\ket{1,+,b}$ is always in $\text{Null}(H(f(s)))$, and there exists a state $\ket{\wt{x}(s)}$ with 
$\ket{\wt{x}(0)} = \ket{0,-,b}, \ket{\wt{x}(1)} = \ket{0,+,x}$, 
such that $\text{Null}(H(f(s)))=\{\ket{\wt{x}(s)},\ket{1,+,b}\}$. 
Such degeneracy will not influence the adiabatic computation starting with 
$\ket{0,-,b}$ for the same reason we discussed for Hermitian 
positive definite case (also discussed in~\cite{SubasiSommaOrsucci2019}), 
and the error of AQC(p) is still bounded by $\eta(s)$ given in Eq.~\eqref{eqn:adiabaticEstimate_eta}. 

Furthermore, the eigenvalue $0$ is separated from the rest of the eigenvalues of $H(f(s))$ by a gap $\Delta(f(s))\ge  \sqrt{(1-f(s))^2+(f(s)/\kappa)^2}$ \cite{SubasiSommaOrsucci2019}. 
For technical simplicity, note that 
$\sqrt{(1-f)^2+(f/\kappa)^2} \geq (1-f+f/\kappa)/\sqrt{2}$ for all $0 \leq f \leq 1$,
we define the lower bound of the gap to be 
$\Delta_*(f) = (1-f+f/\kappa)/\sqrt{2}$, 
which is exactly proportional to that for the Hermitian positive 
definite case. Therefore, we can use exactly the same time schedules as the 
Hermitian positive definite case to perform AQC(p) and AQC(exp) schemes, 
and properties of AQC(p) and AQC(exp) are stated in the following theorems (see Appendices~\ref{app:proof_linear} and~\ref{app:exp} for the proof). 

\begin{thm}\label{thm:main_non_positive}
Let $A\in\CC^{N\times N}$ be a Hermitian matrix (not necessarily positive definite) with condition number $\kappa$. 
For any choice of $1 < p < 2$, the AQC(p) scheme gives
\begin{equation}
    \|\ket{\psi_T(s)}\bra{\psi_T(s)}-\ket{0,+,x}\bra{0,+,x}\|_2 \leq C \kappa/T.
\end{equation}
Therefore, in order to prepare an $\epsilon-$approximation of the solution of QLSP, it suffices to choose the runtime  $T = \Or(\kappa/\epsilon)$. Furthermore, when $p=1,2$, the bound of the runtime becomes $T=\Or(\kappa\log(\kappa)/\epsilon)$.
\end{thm}
\begin{thm}\label{thm:main_exp_non_positive}
   Let $A\in\CC^{N\times N}$ be a Hermitian matrix (not necessarily positive definite) with condition number $\kappa$. Then for large enough $T>0$, the final time error 
    $\|\ket{\psi_T(1)}\bra{\psi_T(1)} - \ket{0,+,x}\bra{0,+,x}\|_2$ 
    of the AQC(exp) scheme is bounded by
\begin{equation}
    C\log(\kappa)\exp\left(-C\left(\frac{\kappa\log^2\kappa}{T}\right)^{-\frac{1}{4}}\right).
\end{equation}
Therefore, for any $\kappa>e$, $0<\epsilon<1$, in order to prepare an $\epsilon-$approximation of the solution of QLSP, it suffices to choose the runtime $T = \Or\left(\kappa\log^2(\kappa)\log^4\left(\frac{\log\kappa}{\epsilon}\right)\right)$. 
\end{thm}

For a most general square matrix $A\in\CC^{N\times N}$, following~\cite{HarrowHassidimLloyd2009} we may transform it into the 
Hermitian case at the expense of further doubling the dimension of the Hilbert space. 
Consider an extended QLSP $\mathfrak{A}\ket{\mathfrak{x}} = \ket{\mathfrak{b}}$ 
in dimension $2N$ where 
\begin{equation*}
    \mathfrak{A} = \sigma_+ \otimes A + \sigma_- \otimes A^{\dagger}
    =\left(\begin{array}{cc}
        0 & A \\
        A^\dagger & 0
    \end{array}\right), \quad  \ket{\mathfrak{b}} = \ket{1,b}.
\end{equation*}
Note that $\mathfrak{A}$ is a Hermitian matrix of dimension $2N$, 
with condition number $\kappa$ and $\|\mathfrak{A}\|_2 = 1$, 
and $\ket{\mathfrak{x}} := \ket{1,x}$ solves the extended QLSP. 
Therefore we can directly apply AQC(p) and AQC(exp) for Hermitian matrix $\mathfrak{A}$ 
to prepare an $\epsilon$-approximation of $x$. 
The total dimension of the Hilbert space  becomes 
$8N$ for non-Hermitian matrix $A$ (therefore three ancilla qubits are needed). 

\begin{figure}
    \centering
    \includegraphics[width=0.45\linewidth]{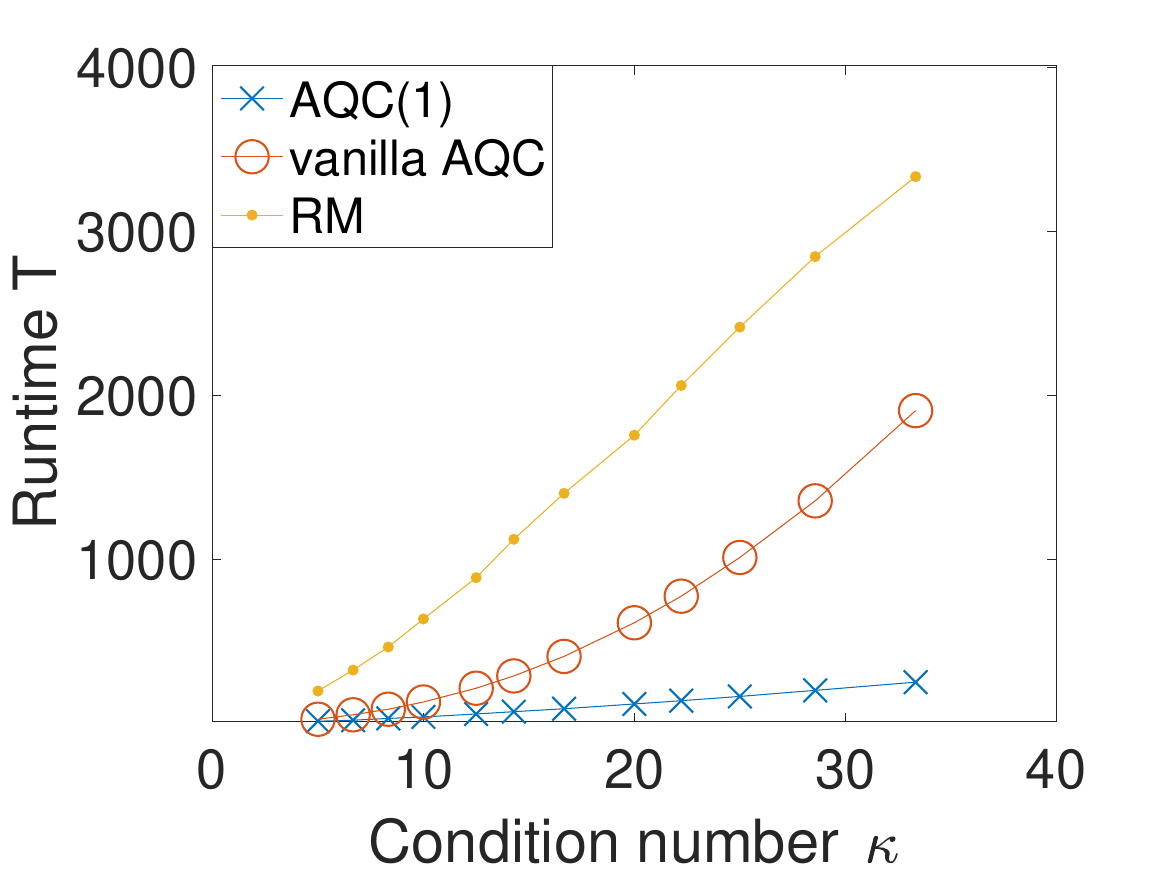}
    \includegraphics[width=0.45\linewidth]{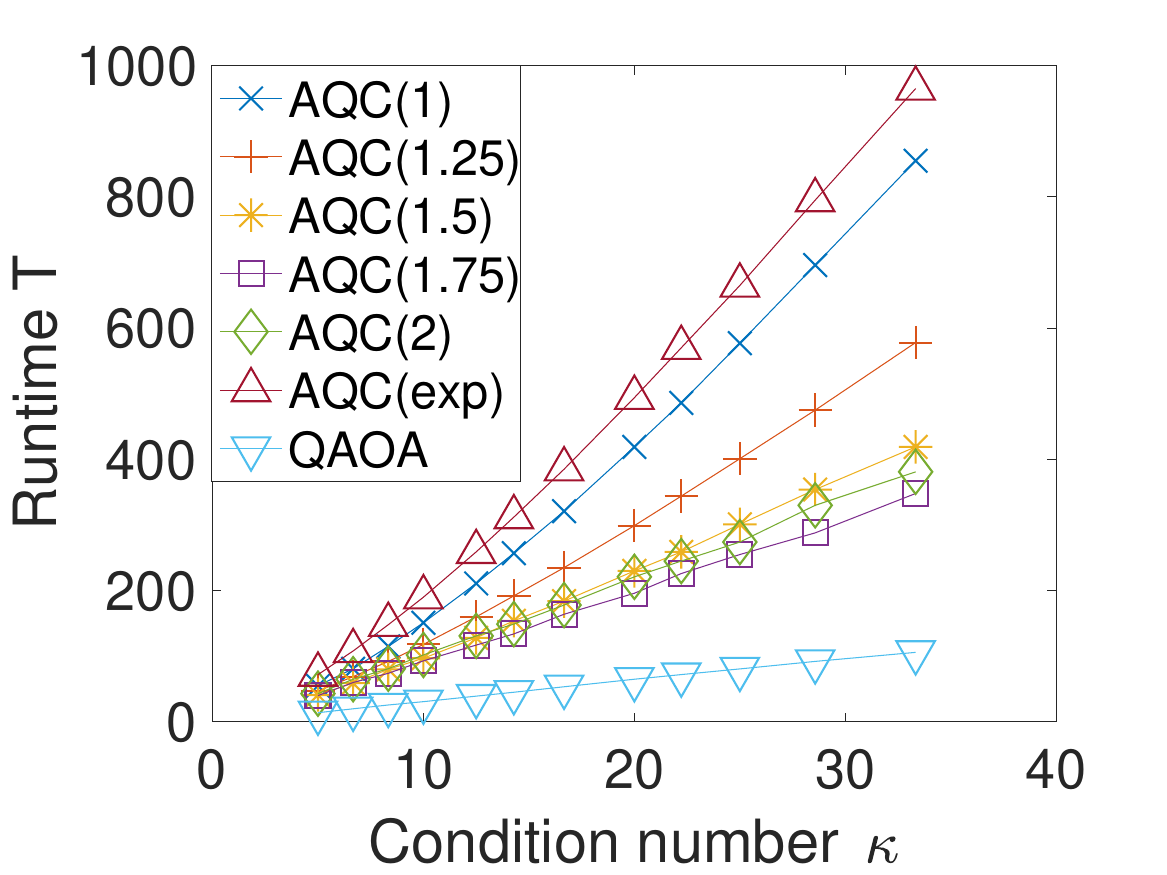}   
    \includegraphics[width=0.6\linewidth]{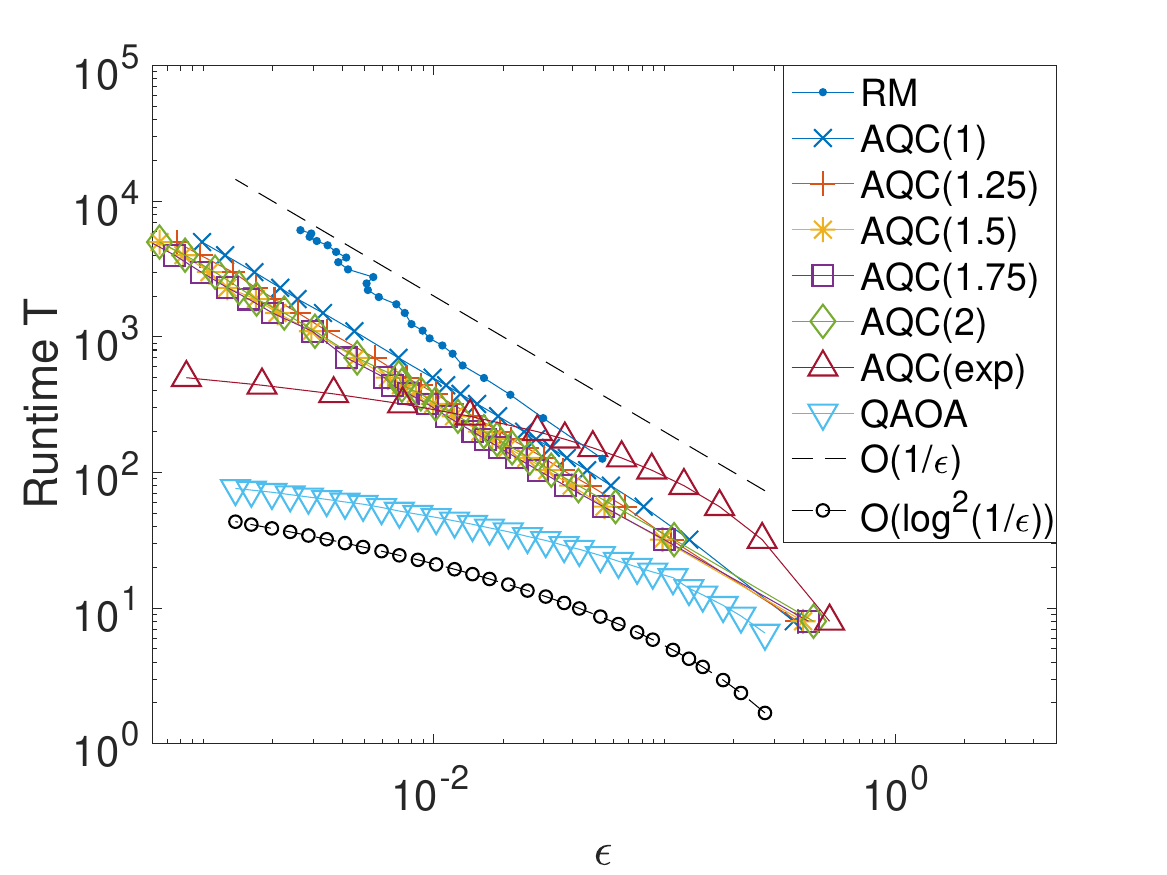}
    \caption{Simulation results for the Hermitian positive definite example. Top (left/right): the runtime to reach desired fidelity $0.99/0.999$ as a function of the condition number. Bottom: a log-log plot of the runtime as a function of the accuracy with  $\kappa = 10$.
} 
    \label{fig:HPD}
\end{figure}
\begin{table}[]
    \centering
    \begin{tabular}{l c c}
        \toprule
        methods & scaling w.r.t. $\kappa$ & scaling w.r.t. $1/\epsilon$ \\\midrule
        vanilla AQC & 2.2022 & / \\
        RM & 1.4912 & 1.3479 \\
        AQC(1) & 1.4619 & 1.0482 \\
        AQC(1.25) & 1.3289 & 1.0248 \\
        AQC(1.5) & 1.2262 & 1.0008\\
        AQC(1.75) & 1.1197 & 0.9899 \\
        AQC(2) & 1.1319 & 0.9904\\
        AQC(exp) & 1.3718 & 0.5377 \\
        AQC(exp) & / & 1.7326 (w.r.t. $\log(1/\epsilon)$) \\
        QAOA & 1.0635 & 0.4188\\
        QAOA & / & 1.4927 (w.r.t. $\log(1/\epsilon)$)\\
        \bottomrule
    \end{tabular}
    \caption{Numerical scaling of the runtime as a function of the condition number and the accuracy, respectively, for the Hermitian positive definite example. }
    \label{tab:HPD_scaling}
\end{table}

\section{Numerical results}\label{sec:numerics}
We first report the performance of AQC(p), AQC(exp), and QAOA for a series of Hermitian positive definite dense matrices with varying condition numbers, together with the performance of RM and vanilla AQC. The details of the setup of the numerical experiments are given in Appendix~\ref{app:numerics}. 

Figure~\ref{fig:HPD} shows how the total runtime $T$ depends on the condition number
$\kappa$ and the accuracy $\epsilon$ for the Hermitian positive definite case. 
The numerical scaling is reported in Table~\ref{tab:HPD_scaling}. For the $\kappa$ dependence, despite that RM and AQC(1) share the same asymptotic linear complexity with respect to $\kappa$, we observe that the preconstant of RM is larger due to its Monte Carlo strategy and the mixed state nature resulting in the same scaling of errors in fidelity and density (see Appendix~\ref{app:RM} for a detailed explanation). 
The asymptotic scaling of the vanilla AQC is at least $\Or(\kappa^2)$. When higher fidelity (0.999) is desired, the cost of vanilla AQC becomes too expensive, and we only report the timing of AQC(p), AQC(exp), and QAOA.  For the $\kappa$ dependence tests, the depth of QAOA ranges from 8 to 60. For the $\epsilon$ dependence test, the depth of QAOA is fixed to be 20. 
We find that the runtime for AQC(p), AQC(exp), and QAOA depends approximately linearly on $\kappa$, while QAOA has the smallest runtime overall. It is also interesting to observe that although the asymptotic scalings of   AQC(1) and AQC(2) are both bounded by $\Or(\kappa \log \kappa)$ instead of $\Or(\kappa)$, the numerical performance of AQC(2) is much better than AQC(1); in fact, the scaling is very close to that with the optimal value of $p$. 
For the $\epsilon$ dependence, the scaling of RM and AQC(p) is $\Or(1/\epsilon)$, which agrees with the error bound. Again the preconstant of RM is slightly larger. 
Our results also confirm that  AQC(exp) only depends poly logarithmically on $\epsilon$. Note that when $\epsilon$ is relatively large, AQC(exp) requires a longer runtime than that of AQC(p), and it eventually outperforms AQC(p) when $\epsilon$ is small enough. 
The numerical scaling of QAOA with respect to $\epsilon$ is found to be only $\Or(\log^{1.5}(1/\epsilon))$ together with the smallest preconstant.

\begin{figure}
    \centering
    \includegraphics[width=0.45\linewidth]{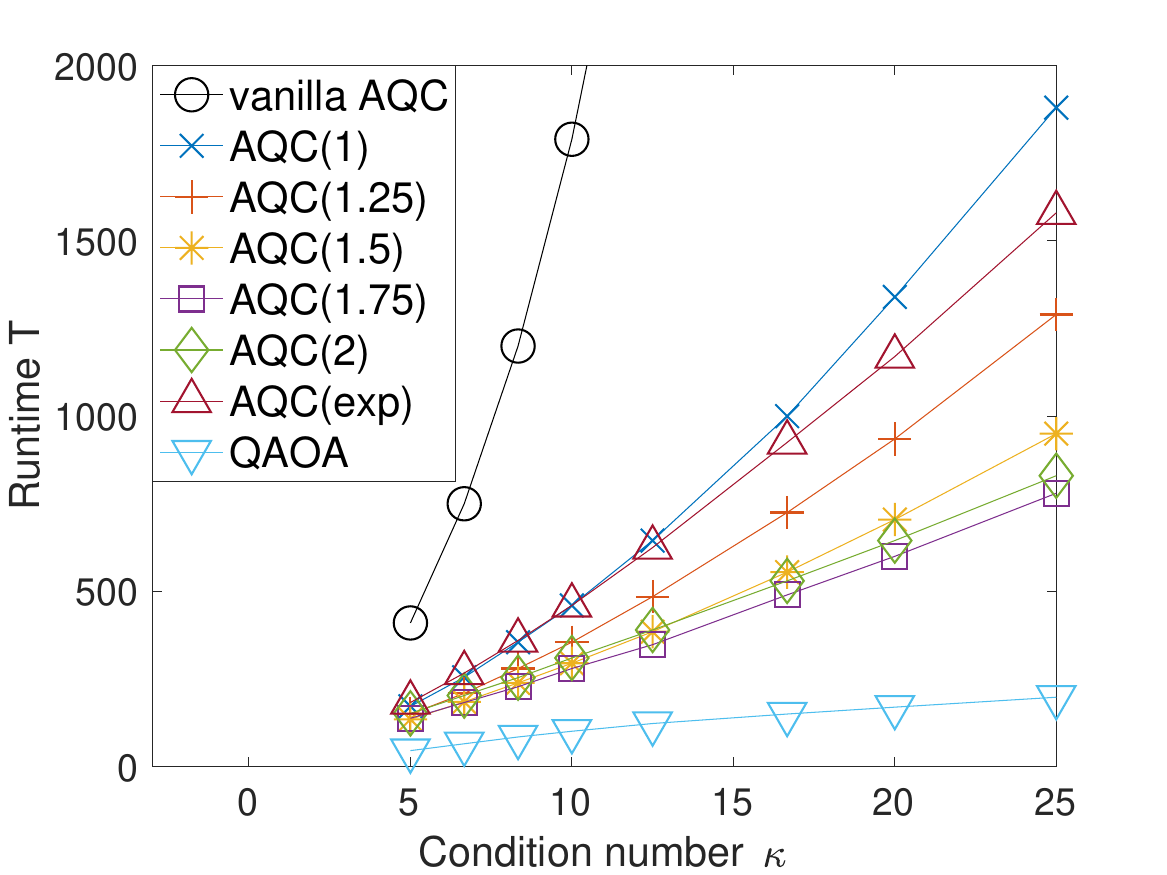}
    \includegraphics[width=0.45\linewidth]{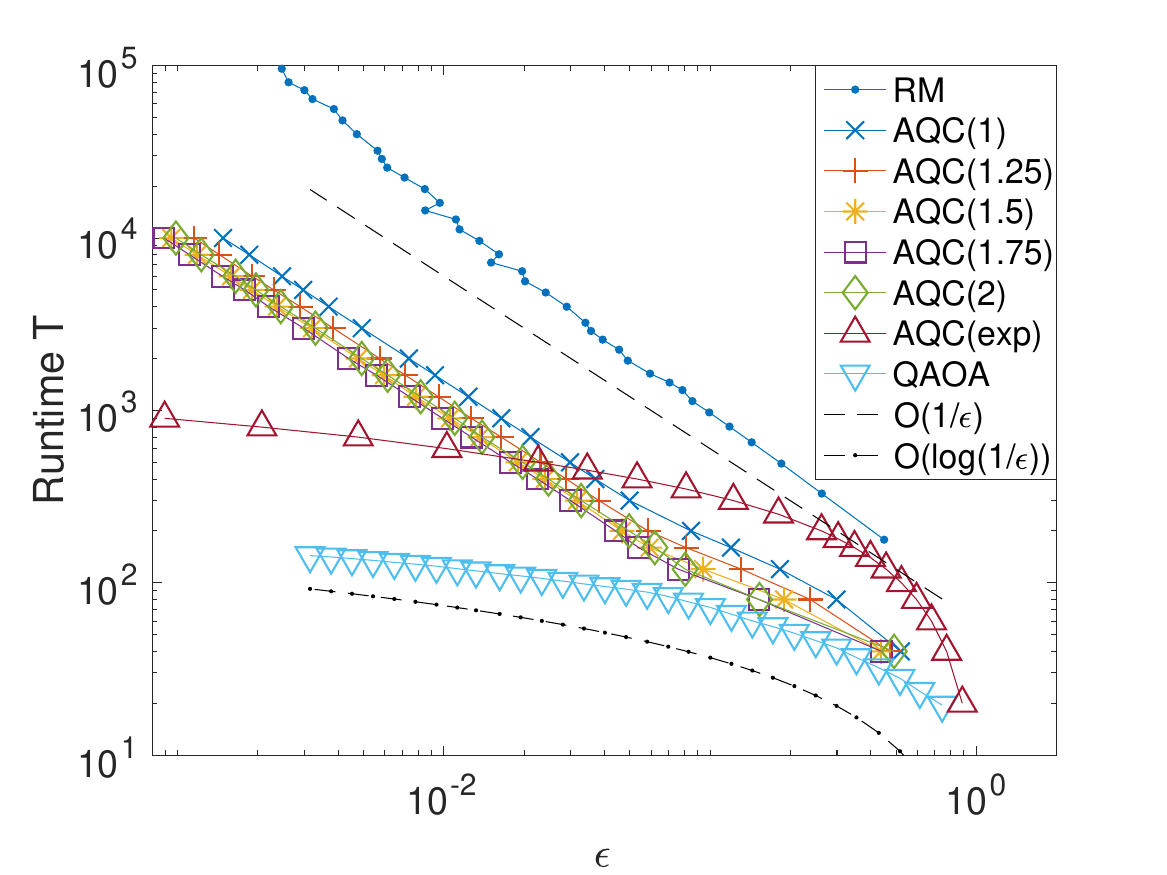}
    \caption{ Simulation results for the non-Hermitian example. Top: the runtime to reach 0.999 fidelity as a function of the condition number. Bottom: a log-log plot of the runtime as a function of the accuracy with $\kappa = 10$.}
    \label{fig:nonH}
\end{figure}
\begin{table}[]
    \centering
    \begin{tabular}{l c c}
        \toprule
        methods & scaling w.r.t. $\kappa$ & scaling w.r.t. $1/\epsilon$ \\\midrule
        vanilla AQC & 2.1980 & / \\
        RM & / & 1.2259 \\
        AQC(1) & 1.4937 & 0.9281 \\
        AQC(1.25) & 1.3485& 0.9274 \\
        AQC(1.5) & 1.2135 & 0.9309 \\
        AQC(1.75) & 1.0790& 0.9378 \\
        AQC(2) & 1.0541 & 0.9425 \\
        AQC(exp) & 1.3438 & 0.4415 \\
        AQC(exp) &  & 0.9316 (w.r.t. $\log(1/\epsilon)$) \\
        QAOA & 0.8907 & 0.3283 \\
        QAOA & / & 0.7410 (w.r.t. $\log(1/\epsilon)$)\\
        \bottomrule
    \end{tabular}
    \caption{Numerical scaling of the runtime as a function of the condition number and the accuracy, respectively, for the non-Hermitian example. }
    \label{tab:nonH_scaling}
\end{table}

Figure~\ref{fig:nonH} and Table~\ref{tab:nonH_scaling} demonstrate the simulation results for non-Hermitian matrices. 
We find that numerical performances of RM, AQC(p), AQC(exp), and QAOA are similar to that of the Hermitian positive definite case.
Again QAOA obtains the optimal performance in terms of the runtime. The numerical scaling of the optimal AQC(p) is found to be $\Or(\kappa/\epsilon)$, while the time complexity of QAOA and AQC(exp) is only  $\Or(\kappa~\text{poly}(\log(\kappa/\epsilon)))$.

\section{Discussion}
By reformulating QLSP into an eigenvalue problem, AQC provides an alternative route to solve QLSP other than those based on phase estimation (such as HHL) and those based on approximation of matrix functions (such as LCU and QSP). However, the scaling of the vanilla AQC is at least $\Or(\kappa^2/\epsilon)$, which is unfavorable with respect to both $\kappa$ and $\epsilon$. Thanks to the explicit information of the energy gap along the adiabatic path, we demonstrate that we may reschedule the AQC and dramatically improve the performance of AQC for solving QLSP. When the target accuracy $\epsilon$ is relatively large, the runtime complexity of the AQC(p) method ($1< p< 2$) is reduced to $\Or(\kappa/\epsilon)$; for highly accurate calculations with a small $\epsilon$, the AQC(exp) method is more advantageous, and its runtime complexity is $\Or(\kappa~\text{poly}(\log(\kappa/\epsilon)))$. 
To our knowledge, our ACP(exp) method provides the first example that an adiabatic algorithm can simultaneously achieve near linear dependence on the spectral gap, and poly-logarithmic dependence on the precision.

Due to the close connection between AQC and QAOA, the runtime complexity of QAOA for solving QLSP is also bounded by $\Or(\kappa~\text{poly}(\log(\kappa/\epsilon)))$. Both AQC and QAOA can be implemented on gate-based quantum computers.  
Our numerical results can be summarized using the following relation:
$$
\text{QAOA}\lesssim \text{AQC(exp)}\lesssim \text{AQC}(p)<\text{RM}<\text{vanilla AQC}.
$$
Here $A<B$ means that the runtime of $A$ is smaller than that of $B$. The two exceptions are: $\text{QAOA}\lesssim \text{AQC(exp)}$ means that the runtime of QAOA is smaller only when the optimizer $\theta$ is found, while $\text{AQC(exp)}\lesssim \text{AQC}(p)$ holds only when $\epsilon$ is sufficiently small. While the runtime complexity of AQC(exp) readily provides an upper bound of the runtime complexity of QAOA, numerical results indicate that the optimizer of QAOA often involves a much smaller depth, and hence the dynamics of QAOA does not necessarily follow the adiabatic path. Therefore, it is of interest to find alternative routes to directly prove the scaling of the QAOA  runtime without relying on AQC.  
In the work~\cite{LinTong2019}, our AQC based algorithm has been combined with the eigenvector filtering technique. Ref.~\cite{LinTong2019} also proposed another AQC inspired quantum linear system solver, which is based on the quantum Zeno effect. Both methods can scale linearly in $\kappa$ and logarithmically in $1/\epsilon$. 
We expect our AQC based QLSP solvers may serve as useful subroutines in other quantum algorithms as well.

\begin{acks}
This work was partially supported by the Department of Energy under Grant No. DE-SC0017867, the Quantum Algorithm Teams Program under Grant No. DE-AC02-05CH11231 (L.L.), by a Google Quantum Research Award, and by the NSF Quantum Leap Challenge Institute (QLCI) program through grant number OMA-2016245 (D. A. and L.L.).
We thank Rolando Somma, Yu Tong and Nathan Wiebe for helpful discussions.
\end{acks}

\bibliographystyle{ACM-Reference-Format}
\bibliography{aqcqaoa}

\appendix

\section{The gap of $H(f(s))$ for Hermitian positive definite matrices} \label{app:gap}
The Hamiltonian $H(f)$ can be written in the block matrix form as
\begin{equation}
    {H}(f) = \left(\begin{array}{cc}
        0 & ((1-f)I + fA)Q_b \\
        Q_b((1-f)I + fA) & 0
    \end{array}\right) {.}
\end{equation}
Let $\lambda$ be an eigenvalue of $H$, then 
\begin{align*}
    0 &= \det \left(\begin{array}{cc}
        \lambda I & -((1-f)I + fA)Q_b \\
        -Q_b((1-f)I + fA) & \lambda I
    \end{array}\right) \\
    &= \det\left(\lambda^2 I - ((1-f)I + fA)Q_b^2((1-f)I + fA)\right)
\end{align*}
where the second equality holds because the bottom two blocks are commutable. 
Thus $\lambda^2$ is an eigenvalue of $((1-f)I + fA)Q_b^2((1-f)I + fA)$, 
and $\Delta^2(f)$ equals the smallest non-zero eigenvalue 
of $((1-f)I + fA)Q_b^2((1-f)I + fA)$. 
Applying a proposition of matrices that $XY$ and $YX$ have the same non-zero 
eigenvalues, $\Delta^2(f)$ also equals the smallest non-zero eigenvalue 
of $Q_b((1-f)I + fA)^2Q_b$.

Now we focus on the matrix $Q_b((1-f)I + fA)^2Q_b$. 
Note that 
$\left|b\right>$ is the unique eigenstate corresponding to the eigenvalue 0, and all 
eigenstates corresponding to non-zero eigenvalues must be orthogonal to 
$\left|b\right>$. Therefore
\begin{align*}
    {\Delta}^2(f) &= \inf_{\left<b|\varphi\right> = 0, \left<\varphi|\varphi\right> = 1} \left<\varphi\left|Q_b((1-f)I + fA)^2Q_b\right|\varphi\right>  \\
    &= \inf_{\left<b|\varphi\right> = 0, \left<\varphi|\varphi\right> = 1} \left<\varphi\left|((1-f)I + fA)^2\right|\varphi\right> \\
    &\geq \inf_{\left<\varphi|\varphi\right> = 1} \left<\varphi\left|((1-f)I + fA)^2\right|\varphi\right> \\
    &= (1-f+f/\kappa)^2 {,}
\end{align*}
and ${\Delta}(f) \geq {\Delta}_*(f) = 1-f+f/\kappa$.

\section{Relations among Different Measurements of Accuracy}\label{app:measurements}

We  show two  relations that connect the error of density 
matrix distance and the error of fidelity, which are used in our proof for AQC(p) and AQC(exp). 
Following the notations in the main text, 
let $\ket{\wt{x}(s)}$ denote the desired 
eigenpath of $H(f(s))$ corresponding to the $0$ eigenvalue, 
and $\text{Null}(H(f(s)))=\{\ket{\wt{x}(s)},\ket{\bar{b}}\}$. 
$P_0(s)$ denotes the projector onto the eigenspace corresponding to the $0$ eigenvalue.

\begin{lem}\label{lem:err_fed}
    (i) The following equation holds, 
    \begin{equation}
        |1-\braket{\psi_T(s)|P_0(s)|\psi_T(s)}| =  1-\left|\braket{\psi_T(s)|\wt{x}(s)}\right|^2 = \|\ket{\psi_T(s)}\bra{\psi_T(s)} - \ket{\wt{x}(s)}\bra{\wt{x}(s)}\|_2^2. 
    \end{equation}
    
    (ii) Assume that 
    $$|1-\braket{\psi_T(s)|P_0(s)|\psi_T(s)}|\le \eta^2(s).$$
    Then the fidelity can be bounded from below by $1-\eta^2(s)$, 
    and the 2-norm error of the density matrix can be bounded from above 
    by $\eta(s)$. 
\end{lem}
\begin{proof}
It suffices only to prove part (i). 
Note that $\ket{\bar{b}}$ is the eigenstate for both $H_0$ and $H_1$ corresponding 
the 0 eigenvalue, we have $H(f(s))\ket{\bar{b}}=((1-f(s))H_0+f(s)H_1)\ket{\bar{b}}=0$, and thus 
$\frac{d}{ds} \braket{\bar{b}|\psi_T(s)} = 0$. 
Together with the initial condition $\braket{\bar{b}|\psi_T(0)}=0$, 
the overlap of $\ket{\bar{b}}$ and $\ket{\psi_T(s)}$ remains to be 0 for the whole time 
period, \ie $\braket{\bar{b}|\psi_T(s)}=0.$
Since $P_0(s)=\ket{\wt{x}(s)}\bra{\wt{x}(s)}+\ket{\bar{b}}\bra{\bar{b}}$,  we have 
$P_0(s)\ket{\psi_T(s)}=\ket{\wt{x}(s)}\braket{\wt{x}(s)|\psi_T(s))}$.
Therefore
$$|1-\braket{\psi_T(s)|P_0(s)|\psi_T(s)}| = |1-\braket{\psi_T(s)|\wt{x}(s)}\braket{\wt{x}(s)|\psi_T(s)}| = 1-\left|\braket{\psi_T(s)|\wt{x}(s)}\right|^2.$$ 

To prove the second equation, let $M = \ket{\psi_T(s)}\bra{\psi_T(s)} - \ket{\wt{x}(s)}\bra{\wt{x}(s)}$. Note that $\|M\|_2^2 = \lambda_{\max}(M^\dagger M)$, we study the eigenvalues of $M^\dagger M$ by first computing that 
$$M^\dagger M = \ket{\psi_T(s)}\bra{\psi_T(s)} + \ket{\wt{x}(s)}\bra{\wt{x}(s)} - \braket{\psi_T(s)|\wt{x}(s)}\ket{\psi_T(s)}\bra{\wt{x}(s)} - \braket{\wt{x}(s)|\psi_T(s)}\ket{\wt{x}(s)}\bra{\psi_T(s)}.$$
Since for any $\ket{y} \in \text{span}\{\ket{\psi_T(s)},\ket{\wt{x}(s)}\}^{\bot}$, $M^\dagger M\ket{y} = 0$, and 
\begin{align*}
    M^\dagger M \ket{\psi_T(s)} &= (1-\left|\braket{\psi_T(s)|\wt{x}(s)}\right|^2)\ket{\psi_T(s)}, \\
    M^\dagger M \ket{\wt{x}(s)} &= (1-\left|\braket{\psi_T(s)|\wt{x}(s)}\right|^2)\ket{\wt{x}(s)}, 
\end{align*}
we have $\|M\|_2^2 = \lambda_{\max}(M^{\dagger}M) = 1-\left|\braket{\psi_T(s)|\wt{x}(s)}\right|^2$. 

\end{proof}

\section{Difference between the scalings of AQC(p) and RM with respect to infidelity}\label{app:RM}

In our numerical test, we observe that to reach a desired fidelity, RM encounters 
a much larger pre-constant than AQC(p). This is due to the following reason. Although the runtime of both RM and AQC(p) scales as $\Or(1/\epsilon)$ where $\epsilon$ is the 2-norm 
error of the density matrix, the scalings with respect to the infidelity are different. More specifically, Lemma~\ref{lem:err_fed} shows that for AQC, the square of the 2-norm error is exactly equal to the infidelity. Thus in order to reach infidelity $1-F$ using AQC(p), it suffices to choose the runtime to be $T = \Or(\kappa/\sqrt{1-F})$. Meanwhile, it has been proved in~\cite{BoixoKnillSomma2009} that the runtime complexity of RM is 
$\wt{\Or}(\kappa/(1-F))$. Therefore, to reach a desired fidelity, the runtime of AQC(p) will be smaller than that of RM, as shown in our numerical examples. 

We further verify the statement above by studying the relation between the 2-norm error of the density matrix and the infidelity 
for AQC(p), AQC(exp) and RM using the positive definite example with 
$\kappa = 10$. In AQC(p) and AQC(exp), we change the runtime to obtain 
approximations with different errors and infidelity. 
In RM we vary the number of exponential operators to obtain different levels of accuracy. 
All other numerical treatments remain unchanged. 
As shown in Figure~\ref{fig:err_fed}, the infidelity is exactly the square of 2-norm error in the case of AQC(p) and AQC(exp), while the infidelity of RM only scales  approximately
linearly with respect to 2-norm error. 
This also verifies that although the runtime of both AQC(p) and RM 
scales linearly with respect to $\epsilon$, the runtime of AQC(p)  can be much smaller to reach desired fidelity. 

\begin{figure}[ht]
    \centering
    \includegraphics[width=0.6\linewidth]{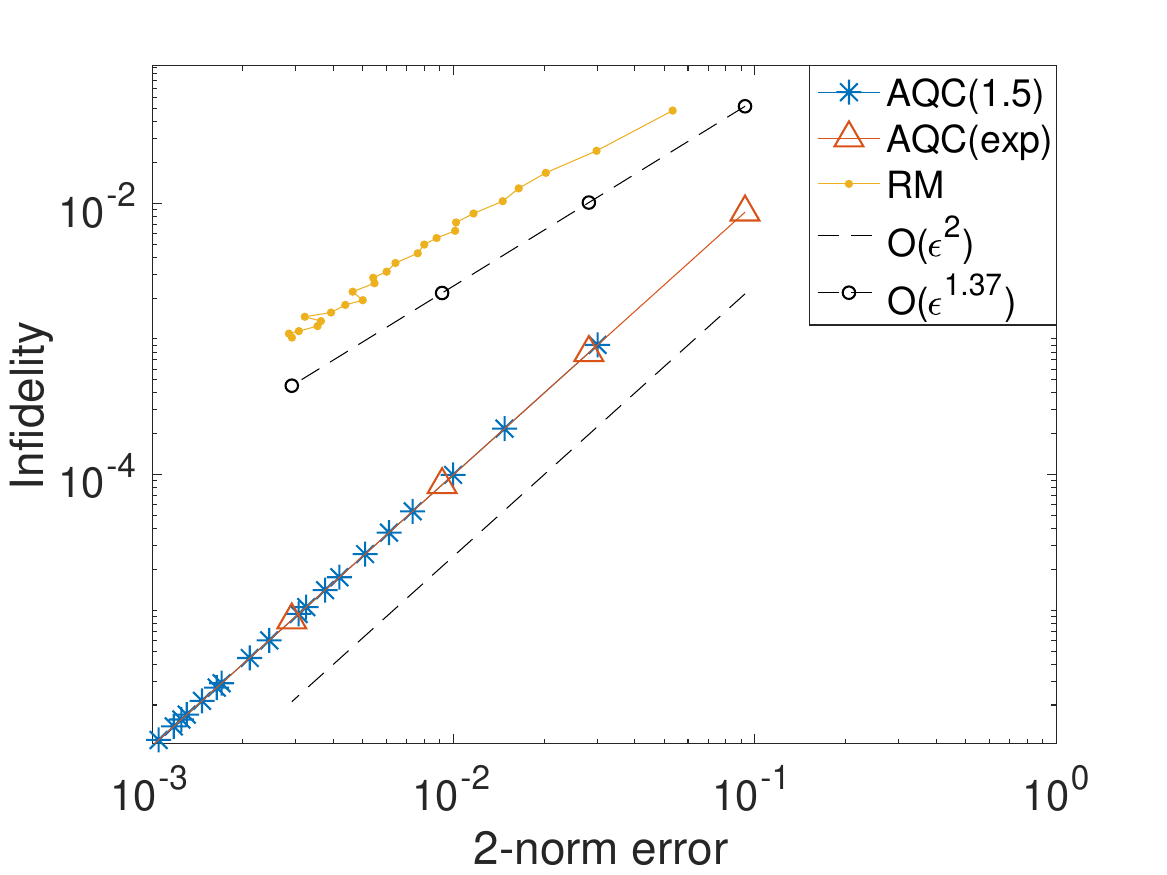}
    \caption{Relation between 2-norm error and infidelity of AQC and RM. }
    \label{fig:err_fed}
\end{figure}

\section{Proof of Theorem~\ref{thm:main} and Theorem~\ref{thm:main_non_positive}}\label{app:proof_linear}

The proof of Theorem~\ref{thm:main} and Theorem~\ref{thm:main_non_positive} rests on some delicate cancellation of the time derivatives $\norm{H^{(1)}}_2,\norm{H^{(2)}}_2$ and the gap $\Delta(f(s))$ in the error bound, 
and can be completed by 
carefully analyzing the $\kappa$-dependence of each term in $\eta(s)$ given
in Eq.~\eqref{eqn:adiabaticEstimate_eta}. 
Note that in both cases ${H}(f) = (1-f)H_0+fH_1$, and 
we let $\Delta_*(f) = (1-f+f/\kappa)/\sqrt{2}$ since such a choice of $\Delta_*$ can 
serve as a lower bound of the spectrum gap for both the case of Theorem~\ref{thm:main} and Theorem~\ref{thm:main_non_positive}. 
We first compute the derivatives of $H(f(s))$ by chain rule as 
\begin{align*}
    H^{(1)}(s) = \frac{d}{ds}{H}(f(s)) 
    = \frac{d {H}(f(s))}{df}\frac{df(s)}{ds}
    = (H_1-H_0)c_p {\Delta}_*^p(f(s)) {,}
\end{align*}
and
\begin{align*}
    H^{(2)}(s) &= \frac{d}{ds}H^{(1)}(s) 
    = \frac{d}{ds}\left((H_1-H_0)c_p {\Delta}_*^p(f(s))\right) \\
    &= (H_1-H_0)c_pp{\Delta}_*^{p-1}(f(s))\frac{d{\Delta}_*(f(s))}{df}\frac{df(s)}{ds} \\
    &= \frac{1}{\sqrt{2}}(-1+1/\kappa)(H_1-H_0)c_p^2p{\Delta}_*^{2p-1}(f(s)){.}
\end{align*}
Then the first two terms of $\eta(s)$ can be rewritten as 
\begin{align*}
    & \frac{\|H^{(1)}(0)\|_2}{T \Delta^2(0)} + \frac{\|H^{(1)}(s)\|_2}{T \Delta^2(f(s))} 
    \leq  \frac{\|H^{(1)}(0)\|_2}{T {\Delta}_*^2(0)} + \frac{\|H^{(1)}(s)\|_2}{T {\Delta}_*^2(f(s))} \\
    = & \frac{\|(H_1-H_0)c_p {\Delta}_*^p(f(0))\|_2}{T {\Delta}_*^2(0)} + \frac{\|(H_1-H_0)c_p {\Delta}_*^p(f(s))\|_2}{T {\Delta}_*^2(f(s))} \\
    \leq & \frac{C}{T}\left(c_p{\Delta}_*^{p-2}(0) + c_p{\Delta}_*^{p-2}(f(s))\right) \\
    \leq & \frac{C}{T}\left(c_p{\Delta}_*^{p-2}(0) + c_p{\Delta}_*^{p-2}(1)\right)
\end{align*}
Here $C$ stands for a general positive constant independent of $s,\Delta,T$. 
To compute the remaining two terms of $\eta(s)$, we use 
the following change of variable 
$$u = f(s'), \quad du = \frac{d}{ds'}f(s')ds' = c_p {\Delta}_*^p(f(s'))ds',$$ 
and the last two terms of $\eta(s)$ become 
\begin{align*}
    &\frac{1}{T}\int_0^s \frac{\|H^{(2)}\|_2}{\Delta^2}ds' 
    \leq  \frac{1}{T}\int_0^s \frac{\|H^{(2)}\|_2}{\Delta_*^2}ds' \\
    = &\frac{1}{T}\int_0^s \frac{\|\frac{1}{\sqrt{2}}(-1+1/\kappa)(H_1-H_0)c_p^2p{\Delta}_*^{2p-1}(f(s'))\|_2}{{\Delta}_*^2(f(s'))}ds' \\
    = & \frac{1}{T}\int_0^{f(s)} \frac{\|\frac{1}{\sqrt{2}}(-1+1/\kappa)(H_1-H_0)c_p^2p{\Delta}_*^{2p-1}(u)\|_2}{{\Delta}_*^2(u)}\frac{du}{c_p{\Delta}_*^p(u)} \\
    \leq & \frac{C}{T}\left((1-1/\kappa)c_p\int_0^{f(s)}{\Delta}_*^{p-3}(u) du\right) \\
    \leq & \frac{C}{T}\left((1-1/\kappa)c_p\int_0^{1}{\Delta}_*^{p-3}(u) du\right){,}
\end{align*}
and similarly
\begin{align*}
    &\frac{1}{T}\int_0^s \frac{\|H^{(1)}\|^2_2}{\Delta^3}ds' 
    \leq  \frac{1}{T}\int_0^s \frac{\|H^{(1)}\|^2_2}{\Delta_*^3}ds' \\
    = & \frac{1}{T}\int_0^s \frac{\|(H_1-H_0)c_p {\Delta}_*^p(f(s'))\|^2_2}{{\Delta}_*^3(f(s'))}ds' \\
    =& \frac{1}{T}\int_0^{f(s)} \frac{\|(H_1-H_0)c_p {\Delta}_*^p(u)\|^2_2}{{\Delta}_*^3(u)}\frac{du}{c_p{\Delta}_*^p(u)} \\
    \leq & \frac{C}{T}\left(c_p \int_0^{f(s)}{\Delta}_*^{p-3}(u) du\right) \\
    \leq & \frac{C}{T}\left(c_p \int_0^{1}{\Delta}_*^{p-3}(u) du\right){.}
\end{align*}
Summarize all terms above, an upper bound of $\eta(s)$ is
\begin{align*}
  \eta(s) &\leq \frac{C}{T}\Big\{\left(c_p{\Delta}_*^{p-2}(0)+ c_p{\Delta}_*^{p-2}(1)\right) + \left((1-1/\kappa)c_p\int_0^{1}{\Delta}_*^{p-3}(u) du\right) + \left(c_p \int_0^{1}{\Delta}_*^{p-3}(u) du\right)\Big\} \\
  &= \frac{C}{T}\Big\{2^{-(p-2)/2}\left(c_p+ c_p\kappa^{2-p}\right) + \left((1-1/\kappa)c_p\int_0^{1}{\Delta}_*^{p-3}(u) du\right) + \left(c_p \int_0^{1}{\Delta}_*^{p-3}(u) du\right)\Big\}{.}
\end{align*}
Finally, since for $1<p<2$ 
\[
c_p = \int_0^1 {\Delta}_*^{-p}(u) du = \frac{2^{p/2}}{p-1}\frac{\kappa}{\kappa-1}(\kappa^{p-1}-1),
\]
and
\begin{align*}
    \int_0^{1}{\Delta}_*^{p-3}(u) du = \frac{2^{-(p-3)/2}}{2-p} \frac{\kappa}{\kappa-1} (\kappa^{2-p}-1) {,}
\end{align*}
we have 
\begin{align*}
    \eta(s)
    \leq &\frac{C}{T}\Big\{\frac{\kappa}{\kappa-1}(\kappa^{p-1}-1) + \frac{\kappa}{\kappa-1}(\kappa-\kappa^{2-p})  \\
    &+ \frac{\kappa}{\kappa-1}(\kappa^{p-1}-1) (\kappa^{2-p}-1) 
    + \left(\frac{\kappa}{\kappa-1}\right)^2(\kappa^{p-1}-1)(\kappa^{2-p}-1)\Big\} {.}
\end{align*}
The leading term of the bound is $\Or(\kappa/T)$  when $1<p<2$. 

Now we consider the limiting case when $p = 1,2$. Note that the bound for $\eta(s)$ can still be written as
\begin{align*}
  \eta(s) &\leq \frac{C}{T}\Big\{\left(c_p{\Delta}_*^{p-2}(0) + c_p{\Delta}_*^{p-2}(1)\right) + \left((1-1/\kappa)c_p\int_0^{1}{\Delta}_*^{p-3}(u) du\right) + \left(c_p \int_0^{1}{\Delta}_*^{p-3}(u) du\right)\Big\} \\
  & = \frac{C}{T}\Big\{2^{-(p-2)/2}\left(c_p+ c_p\kappa^{2-p}\right) + (1-1/\kappa)c_pc_{3-p} + c_p c_{3-p}\Big\}{.}
\end{align*}
Straightforward computation shows that 
$$
c_1 = \int_0^1 {\Delta}_*^{-1}(u) du = \sqrt{2}\frac{\kappa}{\kappa-1}\log(\kappa) 
$$
and
$$c_2 = \int_0^1 {\Delta}_*^{-2}(u) du = 2\frac{\kappa}{\kappa-1}(\kappa-1){.}$$
Hence when $p=1,2$,
\begin{equation*}
  \eta(s) \leq  \frac{C}{T}\Big\{2^{-(p-2)/2}\left(c_p+ c_p\kappa^{2-p}\right) + (1-1/\kappa)c_1c_2 + c_1 c_2\Big\} \leq C\frac{\kappa \log(\kappa)}{T}.
\end{equation*}
This completes the proof of Theorem~\ref{thm:main} and Theorem~\ref{thm:main_non_positive}.

\section{Proof of Theorem~\ref{thm:main_exp} and Theorem~\ref{thm:main_exp_non_positive}}\label{app:exp}

We provide a rigorous proof of the error bound for the AQC(exp) scheme. 
We mainly follow the methodology of~\cite{Nenciu1993} and a part of 
technical treatments of~\cite{GeMolnarCirac2016}. 
Our main contribution is carefully revealing an explicit constant dependence 
in the adiabatic theorem, which is the key to obtain the $\widetilde{\Or}(\kappa)$ scaling. 
In the AQC(exp) scheme, the Hamiltonian 
$H(s) = (1-f(s))H_0 + f(s)H_1$ with $\|H_0\|, \|H_1\| \leq 1$ and 
\begin{equation}\label{eqn-s:scheduling}
    f(s) = \frac{1}{c_e}\int_0^s \exp\left(-\frac{1}{s'(1-s')}\right)\ud s'{.}
\end{equation}
The normalization constant $c_e = \int_0^1 \exp(-\frac{1}{t(1-t)})dt \approx 0.0070$. 
Let $U_T(s)$ denote the corresponding unitary evolution operator, and $P_0(s)$ denote the projector onto the eigenspace corresponding to 0. 
We use $\Delta_*(f) = (1-f+f/\kappa)/\sqrt{2}$ since this can serve as a 
lower bound of the spectrum gap for both the cases of Theorem~\ref{thm:main_exp} and Theorem~\ref{thm:main_exp_non_positive}. 

We first restate the theorems universally with more technical details as 
following. 

\begin{thm}
    Assume the condition number $\kappa > e$. Then the final time adiabatic error $|1-\braket{\psi_T(1)|P_0(1)|\psi_T(1)}|$ of AQC(exp) can be bounded by $\eta_1^2$ where
    
    (a) for arbitrary $N$,  
    \begin{equation*}
        \eta_1^2 = A_1D\log^2\kappa\left( C_2\frac{\kappa\log^2\kappa}{T} N^4\right)^{N}
    \end{equation*}
    where $A_1,D$, and $C_2$ are positive constants which are independent of $T$, $\kappa$ and $N$, 
    
    (b) if $T$ is large enough such that
    \begin{equation*}
       16eA_1^{-1}D\left(\frac{4\pi^2}{3}\right)^3\frac{\kappa\log^2\kappa}{T} \leq 1,
    \end{equation*}
    then 
    \begin{equation*}
        \eta_1^2 = C_1\log^2\kappa\exp\left(-\left(C_2\frac{\kappa\log^2\kappa}{T}\right)^{-\frac{1}{4}}\right) 
    \end{equation*}
    where $A_1,D,C_1$, and $C_2$ are positive constants which are independent of $T$ and $\kappa$. 
    
\end{thm}
\begin{cor}
    For any $\kappa > e, 0<\epsilon<1$, to prepare an $\epsilon$-approximation of the solution of QLSP using AQC(exp), 
    it is sufficient to choose the runtime 
    $T = \Or\left(\kappa\log^2\kappa\log^4\left(\frac{\log\kappa}{\epsilon}\right)\right)$. 
\end{cor}
\begin{proof}
    We start the proof by considering the projector $P(s)$ onto an invariant space of $H$, then $P(s)$ satisfies 
    \begin{equation}
        \I \frac{1}{T} \partial_sP(s) = [H(s),P(s)], \quad P^2(s) = P(s) {.}
    \end{equation}
    We try the ansatz (only formally)
    \begin{equation}\label{eqn:ansatz}
        P(s) = \sum_{j=0}^{\infty} E_{j}(s)T^{-j}.
    \end{equation}
    Substitute it into the Heisenberg equation and match terms with the same orders, we get 
    \begin{equation}\label{eqn:E_Def}
        [H(s),E_0(s)] = 0, \quad \I\partial_s E_j(s) = [H(s), E_{j+1}(s)], \quad E_j(s) = \sum_{m=0}^jE_m(s)E_{j-m}(s) {.}
    \end{equation}
    It has been proved in~\cite{Nenciu1993} that the solution of~\eqref{eqn:E_Def} with initial condition $E_0 = P_0$ is given by
    \begin{align}
        &E_0(s) = P_0(s) = -(2\pi \I)^{-1} \oint_{\Gamma(s)}(H(s)-z)^{-1}dz{,}\label{eqn:E_solu_1} \\
        &E_{j}(s) = (2\pi)^{-1}\oint_{\Gamma(s)}(H(s)-z)^{-1}[E_{j-1}^{(1)}(s),P_0(s)](H(s)-z)^{-1}dz + S_j(s) - 2P_0(s)S_j(s)P_0(s)\label{eqn:E_solu_2}
    \end{align}
    where $\Gamma(s) = \{z\in\CC: |z| = \Delta(s)/2\}$ and
    \begin{equation}\label{eqn:S_def}
        S_j(s) = \sum_{m=1}^{j-1}E_m(s)E_{j-m}(s){.}
    \end{equation}
    Furthermore given $E_0=P_0$, such a solution is unique. 
    
    In general, Eq.~\eqref{eqn:ansatz} does not converge, so for arbitrary positive integer $N$ we define a truncated series as 
    \begin{equation}
        P_{N}(s) = \sum_{j=0}^{N}E_{j}(s)T^{-j}{.}
    \end{equation}
    Then 
    \begin{align*}
        &\I \frac{1}{T} P_{N}^{(1)} - [H,P_{N}]
        = \I\frac{1}{T} \sum_{j=0}^{N}E_{j}^{(1)}T^{-j} - \sum_{j=0}^{N}[H,E_{j}]T^{-j}
        =\I T^{-(N+1)}E_{N}^{(1)}{.}
    \end{align*}
    In Lemma~\ref{lem:E_vanishing_boundary}, we prove that $P_{N}(0)=P_0(0)$ and 
    $P_{N}(1) = P_0(1)$, then the adiabatic error becomes
    \begin{align*}
        |1-\braket{\psi_T(1)|P_0(1)|\psi_T(1)}| 
        &= |\braket{\psi_T(0)|P_0(0)|\psi_T(0)}-\braket{\psi_T(0)|U_T(1)^{-1}P_0(1)U_T(1)|\psi_T(0)}|\\
        &\leq \|P_0(1)-U_{T}(1)^{-1}P_0(0)U_{T}(1)\| \\
        &= \|P_N(1)-U_{T}(1)^{-1}P_N(0)U_{T}(1)\| \\
        &= \left\|\int_{0}^1ds\frac{d}{ds}\left(U_T^{-1}P_{N}U_T\right)\right\| {.}
    \end{align*}
    Straightforward computations show that
    \begin{align*}
        \frac{d}{ds}(U_T^{-1}) = -U_T^{-1}\frac{d}{ds}(U_T)U_T^{-1} 
        = -U_T^{-1}\frac{T}{\I} HU_TU_T^{-1} = -\frac{T}{\I} U_T^{-1}H{,}
    \end{align*}
    \begin{align*}
        \frac{d}{ds}\left(U_T^{-1}P_{N}U_T\right) 
        &= \frac{d}{ds}(U_T^{-1})P_{N}U_T + U_T^{-1}\frac{d}{ds}(P_{N})U_T + U_T^{-1}P_{N}\frac{d}{ds}(U_T) \\
        &= -\frac{T}{\I}U_T^{-1}HP_{N}U_T + 
        U_T^{-1}\frac{T}{\I}[H,P_{N}]U_T + U_T^{-1}T^{-N}E_{N}^{(1)}U_T + \frac{T}{\I}U_T^{-1}P_{N}HU_T \\
        &= T^{-N}U_T^{-1}E_{N}^{(1)}U_T{,}
    \end{align*}
    therefore
    \begin{align*}
        |1-\braket{\psi_T(1)|P_0(1)|\psi_T(1)}|
        \leq \left\|\int_{0}^1T^{-N}U_T^{-1}E_{N}^{(1)}U_Tds\right\| 
        \leq T^{-N} \max_{s\in[0,1]}\|E_{N}^{(1)}\| {.}
    \end{align*}
    
    In Lemma~\ref{lem:E_growth}, we prove that (the constant $c_f = 4\pi^2/3$)
    \begin{align*}
        \|E_{N}^{(1)}\| &\leq A_1A_2^{N}A_3\frac{[(N+1)!]^4}{(1+1)^2(N+1)^2} \\
        &= \frac{A_1}{4}D\log^2\kappa\left[ A_1^{-1}c_f^3\frac{16}{\Delta}D\log^2\kappa\right]^{N}\frac{[(N+1)!]^4}{(N+1)^2}\\
        & \leq \frac{A_1}{4}D\log^2\kappa\left[ 16A_1^{-1}Dc_f^3\kappa\log^2\kappa\right]^{N}\frac{[(N+1)!]^4}{(N+1)^2}\\
        & \leq A_1D\log^2\kappa\left[ 16A_1^{-1}Dc_f^3\kappa\log^2\kappa N^4\right]^{N}
    \end{align*}
    where the last inequality comes from the fact that $[(N+1)!]^4/(N+1)^2 \leq 4N^{4N}$. This completes the proof of part (a).
    
    When $T$ is large enough, we now choose 
    \begin{equation*}
        N = \left\lfloor \left(16eA_1^{-1}Dc_f^3\frac{\kappa\log^2\kappa}{T}\right)^{-\frac{1}{4}}\right\rfloor \geq 1 {,}
    \end{equation*}
    then
    \begin{align*}
        |1-\braket{\psi_T(1)|P_0(1)|\psi_T(1)}| 
        & \leq  A_1D\log^2\kappa\left[ 16A_1^{-1}Dc_f^3\frac{\kappa\log^2\kappa}{T} N^4\right]^{N} \\
        & \leq A_1D\log^2\kappa\exp\left(-\left(16eA_1^{-1}Dc_f^3\frac{\kappa\log^2\kappa}{T}\right)^{-\frac{1}{4}}\right) {.}
    \end{align*}
    This completes the proof of part (b). 
\end{proof}

The remaining part is devoted to some preliminary results regarding $H, E$ and the technical estimates for the growth of $E_j$. It is worth mentioning in advance that in the proof we will encounter many derivatives taken on a contour integral. In fact all such derivatives taken on a contour integral will not involve derivatives on the contour. 
Specifically, since $(H(s)-z)^{-1}$ is analytic for any $0 < |z| < \Delta(s)$, 
    for any $s_0 \in (0,1)$, there exists a small enough neighborhood $B_{\delta}(s_0)$ such that $\forall s \in B_{\delta}(s_0)$, 
    $\oint_{\Gamma(s)} G(s,(H(s)-z)^{-1})dz = \oint_{\Gamma(s_0)} G(s,(H(s)-z)^{-1})dz$ for any smooth mapping $G$. 
    This means locally the contour integral does not depend on the smooth change of the contour, and thus the derivatives will not involve derivatives on the contour. 
    In the spirit of this trick, we write the resolvent $R(z,s,s_0) = (H(s)-z)^{-1}$ for $0 \leq s \leq 1, 0 \leq s_0 \leq 1, z \in \CC$ and $|z| = \Delta(s_0)/2$ and let $R^{(k)}$ denote the partial derivative with respect to $s$, \ie $\frac{\partial}{\partial s}R(z,s,s_0)$, which means by writing $R^{(k)}$ we only consider the explicit time derivatives brought by $H$.

\begin{lem}
     
     (a) $H(s) \in C^{\infty}$ with $H^{(k)}(0) = H^{(k)}(1) = 0$ for all $k \geq 1$.
     
     (b) There is a gap $\Delta(s) \geq \Delta_*(s) = ((1-f(s))+f(s)/\kappa)/\sqrt{2}$ which separates 0 from the rest of the spectrum.
\end{lem}

The following lemma gives the bound for the derivatives of $H$.  
\begin{lem}\label{lem:H_estimate}
    For every $k \geq 1, 0 < s < 1$, 
    \begin{equation}
        \|H^{(k)}(s)\| \leq b(s)a(s)^k\frac{(k!)^2}{(k+1)^2} {,}
    \end{equation}
    where
    \begin{equation*}
        b(s) = \frac{2e}{c_e}\exp\left(-\frac{1}{s(1-s)}\right)[s(1-s)]^2, \quad a(s) = \left(\frac{2}{s(1-s)}\right)^2 {.}
    \end{equation*}
\end{lem}
\begin{proof}
    We first compute the derivatives of $f$. Let $g(s) = -s(1-s)$ and $h(y) = \exp(1/y)$, then $f'(s) = c_e^{-1}h(g(s))$. By the chain rule of high order derivatives (also known as Fa\`a di Bruno's formula), 
    $$f^{(k+1)}(s) = c_e^{-1}\sum \frac{k!}{m_1!1!^{m_1}m_2!2!^{m_2}\cdots m_k!k!^{m_k}}h^{(m_1+m_2+\cdots+m_k)}(g(s))\prod_{j=1}^k\left(g^{(j)}(s)\right)^{m_j}$$
    where the sum is taken over all $k$-tuples of non-negative integers $(m_1,\cdots,m_k)$ satisfying $\sum_{j=1}^k jm_j = k$. Note that 
    $g^{(j)}(s) = 0$ for $j \geq 3$, and the sum becomes 
    \begin{align*}
        f^{(k+1)}(s) &= c_e^{-1}\sum_{m_1+2m_2=k} \frac{k!}{m_1!1!^{m_1}m_2!2!^{m_2}}h^{(m_1+m_2)}(g(s))\left(g^{(1)}(s)\right)^{m_1}\left(g^{(2)}(s)\right)^{m_2} \\
        &= c_e^{-1}\sum_{m_1+2m_2=k} \frac{k!}{m_1!m_2!2^{m_2}}h^{(m_1+m_2)}(g(s))\left(2s-1\right)^{m_1}2^{m_2} \\
        &= c_e^{-1}\sum_{m_1+2m_2=k} \frac{k!}{m_1!m_2!}h^{(m_1+m_2)}(g(s))\left(2s-1\right)^{m_1}{.}
    \end{align*}
    To compute the derivatives of $h$, we use the chain rule again to get (the sum is over $\sum_{j=1}^m jn_j = m$)
    \begin{align*}
        h^{(m)}(y) &= \sum \frac{m!}{n_1!1!^{n_1}n_2!2!^{n_2}\cdots n_m!m!^{n_m}}\exp(1/y)\prod_{j=1}^m\left(\frac{d^{j}(1/y)}{dy^j}\right)^{n_j} \\
        &= \sum \frac{m!}{n_1!1!^{n_1}n_2!2!^{n_2}\cdots n_m!m!^{n_m}}\exp(1/y)\prod_{j=1}^m\left((-1)^jj!y^{-j-1}\right)^{n_j} \\
        &= \sum \frac{(-1)^m m!}{n_1!n_2!\cdots n_m!}\exp(1/y)y^{-m-\sum n_j}
    \end{align*}
    Since $0 \leq n_j \leq m/j$, the number of tuples $(m_1,\cdots,m_n)$ is less than $(m+1)(m/2+1)(m/3+1)\cdots (m/m+1) = \binom{2m}{m} < 2^{2m}$, so for $0 < y < 1$ and $m \leq k$ we have 
    \begin{equation*}
        |h^{(m)}(y)| \leq 2^{2k}k!\exp(1/y)y^{-2k}.
    \end{equation*}
    Therefore $f^{(k+1)}$ can be bounded as 
    \begin{align*}
        |f^{(k+1)}(s)| &\leq c_e^{-1}\sum_{m_1+2m_2=k} \frac{k!}{m_1!m_2!}2^{2k}k!\exp(-\frac{1}{s(1-s)})\left(\frac{1}{s(1-s)}\right)^{2k}|2s-1|^{m_1} \\
        &\leq c_e^{-1}\exp(-\frac{1}{s(1-s)})\left(\frac{2}{s(1-s)}\right)^{2k}(k!)^2\sum_{m_1\leq k}\frac{1}{m_1!} \\
        &\leq ec_e^{-1}\exp(-\frac{1}{s(1-s)})\left(\frac{2}{s(1-s)}\right)^{2k}(k!)^2.
    \end{align*}
    Substitute $k+1$ by $k$ and for every $k \geq 1$
    \begin{align*}
         |f^{(k)}(s)| &\leq ec_e^{-1}\exp\left(-\frac{1}{s(1-s)}\right)\left(\frac{2}{s(1-s)}\right)^{2(k-1)}((k-1)!)^2 \\
         & \leq 4ec_e^{-1}\exp\left(-\frac{1}{s(1-s)}\right)\left(\frac{2}{s(1-s)}\right)^{2(k-1)}\frac{(k!)^2}{(k+1)^2} {.}
    \end{align*}
    Noting that $\|H_0\| \leq 1, \|H_1\| \leq 1$ and $H^{(k)} = (H_1-H_0)f^{(k)}$, we complete the proof of bounds for $H^{(k)}$. 
\end{proof}

The following result demonstrates that $E_{j}$'s for all $j \geq 1$ vanish on the boundary. 
\begin{lem}\label{lem:E_vanishing_boundary}
    
    (a) For all $k \geq 1$, $E_{0}^{(k)}(0) = P_0^{(k)}(0) = 0, E_{0}^{(k)}(1) =P_0^{(k)}(1) = 0$.
    
    (b) For all $j \geq 1, k \geq 0$, $E_{j}^{(k)}(0) = E_{j}^{(k)}(1) = 0$.
\end{lem}
\begin{proof}
    
    We will repeatedly use the fact that $R^{(k)}(0) = R^{(k)}(1) = 0$. This can be proved by taking the $k$-th order derivative of the equation $(H-z)R = I$ and
    \begin{equation*}
        R^{(k)} = -R\sum_{l=1}^k \binom{k}{l}(H-z)^{(l)}R^{(k-l)} = -R\sum_{l=1}^k \binom{k}{l}H^{(l)}R^{(k-l)}{.}
    \end{equation*}
    
    (a) This is a straightforward result by the definition of $E_0$ and the fact that $R^{(k)}$'s vanish on the boundary.
    
    (b) We prove by induction with respect to $j$. For $j = 1$, Eq.~\eqref{eqn:E_solu_2} tells that 
    \begin{equation*}
        E_{1} = (2\pi)^{-1}\oint_{\Gamma}R[P_0^{(1)},P_0]Rdz{.}
    \end{equation*}
    Therefore each term in the derivatives of $E_{1}$ must involve the derivative of $R$ or the derivative of $P_0$, which means the derivatives of $E_{1}$ much vanish on the boundary. 
    
    Assume the conclusion holds for $<j$, then for $j$, first each term of the derivatives of $S_j$ must involve the derivative of some $E_{m}$ with $m < j$, which means the derivatives of $S_j$ must vanish on the boundary. 
    Furthermore, for the similar reason, Eq.~\eqref{eqn:E_solu_2} tells that the derivatives of $E_{j}$ must vanish on the boundary. 
\end{proof}

Before we process, we recall three technical lemmas introduced  in~\cite{Nenciu1993,GeMolnarCirac2016}. Throughout 
let $c_f = 4\pi^2/3$ denote an absolute constant. 
\begin{lem}\label{lem:tech_1}
    Let $\alpha>0$ be a positive real number, $p,q$ be non-negative integers and $r = p+q$. Then
    \begin{equation*}
        \sum_{l=0}^k \binom{k}{l}\frac{[(l+p)!(k-l+q)!]^{1+\alpha}}{(l+p+1)^2(k-l+q+1)^2} \leq c_f\frac{[(k+r)!]^{1+\alpha}}{(k+r+1)^2}{.}
    \end{equation*}
\end{lem}
\begin{lem}\label{lem:tech_2}
    Let $k$ be a non-negative integer, then 
    \begin{equation*}
        \sum_{l=0}^k\frac{1}{(l+1)^2(k+1-l)^2} \leq c_f\frac{1}{(k+1)^2}{.}
    \end{equation*}
\end{lem}
\begin{lem}\label{lem:tech_3}
    Let $A(s),B(s)$ be two smooth matrix-valued functions defined on $[0,1]$ satisfying
    \begin{equation*}
        \|A^{(k)}(s)\| \leq a_1(s)a_2(s)^k\frac{[(k+p)!]^{1+\alpha}}{(k+1)^2}, \quad 
        \|B^{(k)}(s)\| \leq b_1(s)b_2(s)^k\frac{[(k+q)!]^{1+\alpha}}{(k+1)^2}
    \end{equation*}
    for some non-negative functions $a_1,a_2,b_1,b_2$, non-negative integers $p,q$ and for all $k \geq 0$. Then for every $k \geq 0,0 \leq s \leq 1$, 
    \begin{equation*}
        \|(A(s)B(s))^{(k)}\| \leq c_fa_1(s)b_1(s)\max\{a_2(s),b_2(s)\}^k\frac{[(k+r)!]^{1+\alpha}}{(k+1)^2}
    \end{equation*}
    where $r = p+q$. 
\end{lem}

Next we bound the derivatives of the resolvent. This bound provides 
the most important improvement of the general adiabatic bound. 
\begin{lem}\label{lem:R_estimate}
     For all $k \geq 0$, 
    \begin{equation*}
        \|R^{(k)}(z,s_0,s_0)\| \leq \frac{2}{\Delta(s_0)}\left(D\log^2\kappa\right)^k\frac{(k!)^4}{(k+1)^2} 
    \end{equation*}
    where $$D = c_f\frac{2048\sqrt{2}e^2}{c_e}$$. 
\end{lem}
\begin{proof}
    We prove by induction, and for simplicity we will omit explicit dependence on arguments $z, s$, and $s_0$. The estimate obviously holds for $k = 0$. 
    Assume the estimate holds for $<k$. Take the $k$th order derivative of the equation $(H-z)R = I$ and we get 
    \begin{equation*}
        R^{(k)} = -R\sum_{l=1}^k \binom{k}{l}(H-z)^{(l)}R^{(k-l)} = -R\sum_{l=1}^k \binom{k}{l}H^{(l)}R^{(k-l)}{.}
    \end{equation*}
    Using Lemma~\ref{lem:H_estimate} and the induction hypothesis, we have
    \begin{align*}
        \|R^{(k)}\|_2 &\leq \frac{2}{\Delta}\sum_{l=1}^k\binom{k}{l}ba^l\frac{(l!)^2}{(l+1)^2}\frac{2}{\Delta}\left(D\log^2\kappa\right)^{k-l}\frac{[(k-l)!]^4}{(k-l+1)^2}
    \end{align*}
    
    To proceed we need to bound the term $\Delta^{-1}ba^l$ for $l \geq 1$. 
    Let us define 
    \begin{equation*}
        F(s) = \frac{c_e}{2^{2l}2\sqrt{2}e}\Delta^{-1}_*(s)b(s)a(s)^l = \frac{\exp(-\frac{1}{s(1-s)})}{(1-f(s)+f(s)/\kappa)[s(1-s)]^{2l-2}}{.}
    \end{equation*}
    Note that $F(0)=F(1)=0, F(s) > 0$ for $s\in(0,1)$ and $F(1/2+t)>F(1/2-t)$ for $t \in (0,1/2)$,  
    then there exists a maximizer $s_* \in [1/2,1)$ such that 
    $F(s) \leq F(s_*), \forall s \in [0,1]$. 
    Furthermore, $F'(s_*) = 0$. Now we compute the $F'$ as
    \begin{align*}
        &[(1-f+f/\kappa)[s(1-s)]^{2l-2}]^2F'(s) \\
        = & \exp\left(-\frac{1}{s(1-s)}\right)\frac{1-2s}{s^2(1-s)^2}(1-f+f/\kappa)[s(1-s)]^{2l-2} \\
        & \quad - \exp\left(-\frac{1}{s(1-s)}\right)\left[(-f'+f'/\kappa)[s(1-s)]^{2l-2} + (1-f+f/\kappa)(2l-2)[s(1-s)]^{2l-3}(1-2s)\right] \\
        = & \exp\left(-\frac{1}{s(1-s)}\right)[s(1-s)]^{2l-4} \\
        & \times \Big[ (1-f+f/\kappa)(1-2s)[1-(2l-2)s(1-s)] - \exp\left(-\frac{1}{s(1-s)}\right)c_e^{-1}(-1+1/\kappa)s^2(1-s)^2\Big] \\
        =& \exp\left(-\frac{1}{s(1-s)}\right)[s(1-s)]^{2l-4}G(s)
    \end{align*}
    where
    \begin{equation*}
        G(s)= (1-f+f/\kappa)(1-2s)[1-(2l-2)s(1-s)] + \exp\left(-\frac{1}{s(1-s)}\right)c_e^{-1}(1-1/\kappa)s^2(1-s)^2{.}
    \end{equation*}
    The sign of $F'(s)$ for $s \in (0,1)$ is the same as the sign of $G(s)$. 
    
    We now show that $s_*$ cannot be very close to 1. Precisely, we will prove that for all $s \in [1-\frac{c}{l\log\kappa},1)$ with $c = \sqrt{c_e}/4 \approx 0.021$, $G(s) < 0$. For such $s$, we have 
    \begin{equation*}
        1-f+f/\kappa \geq f(1/2)/\kappa > 0, 
    \end{equation*}
    \begin{equation*}
        1-2s < -1/2, 
    \end{equation*}
    and 
    \begin{equation*}
        1-(2l-2)s(1-s) \geq 1-(2l-2)(1-s) \geq 1-\frac{2c}{\log\kappa} \geq 1/2{,}
    \end{equation*}
    then 
    \begin{align*}
        &(1-f+f/\kappa)(1-2s)[1-(2l-2)s(1-s)] \leq - \frac{f(1/2)}{4\kappa} = -\frac{1}{8\kappa}{.}
    \end{align*}
    On the other hand, 
    \begin{align*}
        \exp\left(-\frac{1}{s(1-s)}\right) 
        & \leq \exp \left(-(1-\frac{c}{l\log\kappa})^{-1}\frac{l\log\kappa}{c}\right)  \\
        & = \kappa^{-(1-\frac{c}{l\log\kappa})^{-1}\frac{l}{c}} \\
        & \leq \kappa^{-l/c} \\
        & \leq \kappa^{-1} {,}
    \end{align*}
    then
    \begin{align*}
        &\exp\left(-\frac{1}{s(1-s)}\right)c_e^{-1}(1-1/\kappa)s^2(1-s)^2 \\
        \leq & \frac{1}{\kappa}\frac{1}{c_e}\left(\frac{c}{l\log\kappa}\right)^2 \\
        \leq & \frac{1}{16\kappa} {.}
    \end{align*}
    Therefore for all $s \in [1-\frac{c}{l\log\kappa},1]$ we have 
    $G(s) \leq -1/(16\kappa) < 0$, which indicates $s_* \leq 1-\frac{c}{l\log\kappa}$.
    
    We are now ready to bound $F(s)$. From the equation $G(s_*) = 0$, we get 
    \begin{equation*}
        \frac{\exp\left(-\frac{1}{s_*(1-s_*)}\right)}{1-f+f/\kappa} = \frac{(1-2s_*)[1-(2l-2)s_*(1-s_*)]}{c_e^{-1}(-1+1/\kappa)s_*^2(1-s_*)^2} {,}
    \end{equation*}
    which gives 
    \begin{align*}
        F(s) &\leq F(s_*) \\
        &= \frac{(1-2s_*)[1-(2l-2)s_*(1-s_*)]}{c_e^{-1}(-1+1/\kappa)[s_*(1-s_*)]^{2l}} \\
        & \leq \frac{2s_*-1}{c_e^{-1}(1-1/\kappa)[s_*(1-s_*)]^{2l}} \\
        & \leq 2c_e\cdot 2^{2l}(1-s_*)^{-2l} \\
        & \leq 2c_e\cdot 2^{2l}\left(\frac{l\log\kappa}{c}\right)^{2l} \\
        & = 2c_e\left(\frac{64}{c_e}\right)^l(\log\kappa)^{2l}l^{2l}\\
        & \leq \frac{2c_e}{e^2}\left(\frac{64e^2}{c_e}\right)^l(\log\kappa)^{2l}(l!)^2{.}
    \end{align*}
    The last inequality comes from the fact $l^l \leq e^{l-1}l!$, which can be derived from the fact that
\begin{displaymath}
\sum_{i=1}^n \log i\ge \int_1^n\log x\ud x=n\log n-(n-1).
\end{displaymath}
 By definition of $F(s)$ we immediately get
    \begin{equation*}
        \Delta^{-1}ba^l \leq \frac{2\sqrt{2}e}{c_e}4^lF \leq \frac{4\sqrt{2}}{e}\left(\frac{256e^2}{c_e}\right)^l(\log\kappa)^{2l}(l!)^2 {.}
    \end{equation*}
    
    Now we go back to the estimate of $R^{(k)}$. By Lemma~\ref{lem:tech_1},  
    \begin{align*}
        \|R^{(k)}\|_2 &\leq \frac{2}{\Delta}\sum_{l=1}^k\binom{k}{l}ba^l\frac{(l!)^2}{(l+1)^2}\frac{2}{\Delta}\left(D\log^2\kappa\right)^{k-l}\frac{[(k-l)!]^4}{(k-l+1)^2} \\
        & \leq \frac{2}{\Delta}\sum_{l=1}^k\binom{k}{l}\frac{8\sqrt{2}}{e}\left(\frac{256e^2}{c_e}\right)^l(\log\kappa)^{2l}(l!)^2\frac{(l!)^2}{(l+1)^2}\left(D\log^2\kappa\right)^{k-l}\frac{[(k-l)!]^4}{(k-l+1)^2} \\
        & \leq \frac{2}{\Delta} (D\log^2\kappa)^kc_f^{-1}\sum_{l=1}^k\binom{k}{l}\frac{(l!)^4[(k-l)!]^4}{(l+1)^2(k-l+1)^2} \\
        & \leq \frac{2}{\Delta}(D\log^2\kappa)^k\frac{(k!)^4}{(k+1)^2}{.}
    \end{align*}
    This completes the proof.
\end{proof}

The next lemma is the main technical result, which gives the bound of derivatives of $E_{j}$ defined in Equation~\eqref{eqn:E_Def}. 
\begin{lem}\label{lem:E_growth}
    
    (a) For all $k \geq 0$, 
    \begin{equation}
        \|E_0^{(k)}\| = \|P_0^{(k)}\| \leq (D\log^2\kappa)^k\frac{(k!)^4}{(k+1)^2}{.}
    \end{equation}
    
    (b) For all $k \geq 0, j \geq 1$, 
    \begin{equation}
        \|E_{j}^{(k)}\| \leq A_1A_2^{j}A_3^{k}\frac{[(k+j)!]^4}{(k+1)^2(j+1)^2}
    \end{equation}
    with
    \begin{align*}
        A_1 &= \frac{1}{2}\left[c_f^2\left(1+2c_f^2\right)\right]^{-1},\\
        A_2 &= A_1^{-1}c_f^3\frac{16}{\Delta}D\log^2\kappa,\\
        A_3 &= D\log^2\kappa{.}
    \end{align*}
\end{lem}
\begin{rem}
    The choice of $A_1,A_2$  can be rewritten as 
    \begin{align*}
        c_f^3\frac{16}{\Delta}D\log^2\kappa & = A_1A_2, \\
        c_f^2\left(1+2c_f^2\right)A_1 &= \frac{1}{2}. 
    \end{align*}
    Furthermore, using $c_f>1$, we have
    \begin{displaymath}
        c_f^3\frac{16}{\Delta}\frac{A_3}{A_2} =A_1\leq \frac{1}{2}{.}
\end{displaymath}
    These relations will be used in the proof later. 
\end{rem}
\begin{proof}

    (a) By Lemma~\ref{lem:R_estimate}, 
    \begin{equation*}
        \|P_0^{(k)}(s_0)\| = \left\|(2\pi \I)^{-1} \oint_{\Gamma(s_0)}R^{(k)}(z,s_0,s_0)dz\right\| 
        \leq (D\log^2\kappa)^k\frac{(k!)^4}{(k+1)^2}
    \end{equation*}
    
    (b) We prove by induction with respect to $j$. For $j = 1$, Eq.~\eqref{eqn:E_solu_2} tells 
    \begin{equation*}
        \|E_1^{(k)}\| = \left\|(2\pi)^{-1}\oint_{\Gamma}\frac{d^k}{ds^k}(R[P_{0}^{(1)},P_0]R)dz\right\| 
        \leq \frac{\Delta}{2}\left\|\frac{d^k}{ds^k}(R[P_{0}^{(1)},P_0]R)\right\|{.}
    \end{equation*}
    By Lemma~\ref{lem:tech_3} and Lemma~\ref{lem:R_estimate}, 
    \begin{align*}
        \|E_1^{(k)}\| &\leq \Delta c_f^3\left(\frac{2}{\Delta}\right)^2D\log^2\kappa(D\log^2\kappa)^k\frac{[(k+1)!]^4}{(k+1)^2} \\
        &\leq A_1A_2A_3^k\frac{[(k+1)!]^4}{(k+1)^2(1+1)^2} {.}
    \end{align*}
    Now assume $<j$ the estimate holds, for $j$, by Lemma~\ref{lem:tech_2}, Lemma~\ref{lem:tech_3} and the induction hypothesis, 
    \begin{align*}
        \|S_j^{(k)}\| &\leq \sum_{m=1}^{j-1}c_fA_1A_2^mA_1A_2^{j-m}A_3^k\frac{[(k+j)!]^4}{(k+1)^2(m+1)^2(j-m+1)^2} \\
        &= A_1^2A_2^jA_3^k\frac{[(k+j)!]^4}{(k+1)^2}c_f\sum_{m=1}^{j-1}\frac{1}{(m+1)^2(j-m+1)^2} \\
        &\leq c_f^2A_1^2A_2^jA_3^k\frac{[(k+j)!]^4}{(k+1)^2(j+1)^2} {.}
    \end{align*}
    Again by Lemma~\ref{lem:tech_3}, Lemma~\ref{lem:R_estimate} and the induction hypothesis, 
    \begin{align*}
        \|E_j^{(k)}\| &\leq \|\frac{d^k}{ds^k}\left((2\pi)^{-1}\oint_{\Gamma}R[E_{j-1}^{(1)},P_0]Rdz\right)\| + \|\frac{d^k}{ds^k}S_j\| + \|\frac{d^k}{ds^k}\left(2P_0S_jP_0\right)\| \\
        &\leq \Delta c_f^3\left(\frac{2}{\Delta}\right)^2A_1A_2^{j-1}A_3
        \frac{1}{j^2}A_3^k\frac{[(k+j)!]^4}{(k+1)^2}  +c_f^2A_1^2A_2^jA_3^k\frac{[(k+j)!]^4}{(k+1)^2(j+1)^2} \\
        & \quad\quad + 2c_f^2c_f^2A_1^2A_2^j\frac{1}{(j+1)^2}A_3^k\frac{[(k+j)!]^4}{(k+1)^2} \\
        & \leq c_f^3\frac{16}{\Delta}A_1A_2^{j-1}A_3^{k+1}\frac{[(k+j)!]^4}{(k+1)^2(j+1)^2}  +c_f^2A_1^2A_2^jA_3^k\frac{[(k+j)!]^4}{(k+1)^2(j+1)^2} \\
        & \quad\quad
        + 2c_f^4A_1^2A_2^jA_3^k\frac{[(k+j)!]^4}{(k+1)^2(j+1)^2} \\
        & = \left[c_f^3\frac{16}{\Delta}\frac{A_3}{A_2} + c_f^2\left(1+2c_f^2\right)A_1\right] \times \left[A_1A_2^jA_3^k\frac{[(k+j)!]^4}{(k+1)^2(j+1)^2}\right] \\
        & \leq A_1A_2^jA_3^k\frac{[(k+j)!]^4}{(k+1)^2(j+1)^2}{.}
    \end{align*}
\end{proof}

\section{Details of the numerical treatments and examples}\label{app:numerics}
For simulation purpose, the AQC schemes are carried out using the first-order Trotter splitting method with a time step size $0.2$. We use the gradient descent method to optimize QAOA and record the running time corresponding to the lowest error in each case. 
In QAOA we also use the true fidelity to measure the error. RM is a Monte Carlo method, and each RM calculation involves performing 200 independent runs to obtain the density matrix $\rho^{(i)}$ for $i$-th repetition, then we use the averaged density $\bar{\rho} = 1/n_{\text{rep}}\sum \rho^{(i)}$ to compute the error. We report the averaged runtime of each single RM calculation. We perform calculations for a series of 64-dimensional Hermitian positive definite dense matrices $A_1$, and 32-dimensional non-Hermitian dense matrices $A_2$ with varying condition number $\kappa$. 

For concreteness, for the Hermitian positive definite example, we choose $A = U\Lambda U^{\dagger}$. 
Here $U$ is an orthogonal matrix obtained by Gram-Schmidt orthogonalization (implemented via a QR factorization) of 
the discretized periodic Laplacian operator  given by 
\begin{equation}
    L = \left(\begin{array}{cccccc}
        1 & -0.5 & & & &-0.5 \\
        -0.5 & 1 & -0.5 & & & \\
         & -0.5 & 1 & -0.5 & & \\
         &  & \ddots &\ddots & \ddots& \\
         &  & & -0.5& 1& -0.5 \\
        -0.5 &  & & & -0.5 & 1 \\
    \end{array}\right){.}
\end{equation}
$\Lambda$ is chosen to be a diagonal matrix with diagonals 
uniformly distributed in $[1/\kappa,1]$. More precisely, 
$\Lambda = \text{diag}(\lambda_1,\lambda_2,\cdots,\lambda_{N})$ with $\lambda_k = 1/\kappa + (k-1)h, h = (1-1/\kappa)/(N-1)$. Such construction ensures $A$ to be 
a Hermitian positive definite matrix which satisfies $\|A\|_2 = 1$ and the condition number of $A$ is $\kappa$. 
We choose $\ket{b} = \sum_{k=1}^{N}u_k / \|\sum_{k=1}^{N}u_k\|_2$ 
where $\{u_k\}$ is the set of the column vectors of $U$.  Here $N=64$.

For the non-Hermitian positive definite example, we choose $A = U\Lambda V^{\dagger}$. 
Here $U$ is the same as those in the Hermitian positive definite case, 
except that the dimension is reduced to $N=32$. 
$\Lambda = \text{diag}(\lambda_1,\lambda_2,\cdots,\lambda_N)$ with 
$\lambda_k = (-1)^k(1/\kappa+(k-1)h), h = (1-1/\kappa)/(N-1)$. 
$V$ is an orthogonal matrix obtained by Gram-Schmidt orthogonalization of 
the matrix
\begin{equation}
    K = \left(\begin{array}{cccccc}
        2 & -0.5 & & & &-0.5 \\
        -0.5 & 2 & -0.5 & & & \\
         & -0.5 & 2 & -0.5 & & \\
         &  & \ddots &\ddots & \ddots& \\
         &  & & -0.5& 2& -0.5 \\
        -0.5 &  & & & -0.5 & 2 \\
    \end{array}\right){.}
\end{equation}
Such construction ensures $A$ to be non-Hermitian, satisfying $\|A\|_2 = 1$ 
and the condition number of $A$ is $\kappa$. 
We choose the same $\ket{b}$ as that in the Hermitian positive definite example.

\end{document}